\documentclass[11pt]{article}

\usepackage{etoolbox}
\usepackage{graphicx}

\newcommand\tr{{\mathpalette\raiseT{\intercal}}} 
\newcommand\raiseT[2]{\raisebox{0.25ex}{$#1#2$}} 

\newenvironment{prooft}[1]{\begin{proof}[Proof of #1]}{\end{proof}}

\usepackage[svgnames]{xcolor}
\usepackage{amsmath, amssymb, amsthm}
\usepackage{mathtools}
\usepackage{stix}
\usepackage{microtype}
\usepackage{xfrac}
\usepackage[margin=1in,left=1in,right=1in,marginpar=0.75in]{geometry}
\usepackage{graphicx}
\usepackage[colorlinks=true,urlcolor=blue,citecolor=blue,linkcolor=blue]{hyperref}
\usepackage{natbib}
\usepackage{setspace}
\usepackage[disable]{todonotes}
\usepackage{cleveref}
\usepackage{soul}
\usepackage{accents}

\newcommand{\older}[1]{}
\usepackage[inline]{enumitem}

\allowdisplaybreaks
\newcommand{\R}[1]{\boldsymbol{#1}}

\newtheoremstyle{jorisstyle}
{3pt}
{3pt}
{}
{}
{\bfseries}
{}
{ }
{\thmname{#1}\thmnumber{ #2'}\thmnote{\bfseries~#3}.}

\DeclareRobustCommand\<[1][black]{\begin{equation}{\color{#1}}}
\DeclareRobustCommand\>{\end{equation}}

\newcommand\ot\leftarrow
\newcommand{\e}[2]{\epsilon_{#1#2}}							
\newcommand{\convd}{\stackrel{d}{\to}}
\theoremstyle{definition}

\theoremstyle{jorisstyle}

\newtheorem*{just*}{Justification}

\theoremstyle{plain}
\newtheorem{thm}{Theorem}
\newcommand{\sB}{\mathscr{B}}

\renewcommand{\e}{\mathrm{e}}
\theoremstyle{plain}
\newtheorem{lem}{Lemma}
\newtheorem*{thm*}{Theorem}

\newtheorem{ex}{Example}

\definecolor{dkgreen}{rgb}{0,0.6,0}
\definecolor{gray}{rgb}{0.5,0.5,0.5}
\definecolor{mauve}{rgb}{0.58,0,0.82}

\newcommand{\Reals}{\mathbb{R}}

\DeclarePairedDelimiter{\parens}{(}{)}
\DeclarePairedDelimiter\cparens\{\}
\DeclarePairedDelimiter\sparens[]

\DeclarePairedDelimiter{\abs}\vert\vert

\providecommand\given{}
\newcommand\Symbol[1][]{%
	\nonscript\:#1\vert
	\allowbreak
	\nonscript\:
	\mathopen{}}
\DeclarePairedDelimiterX\condr[1](){\renewcommand\given{\Symbol[\delimsize]}#1}
\DeclarePairedDelimiterX\condc[1]\{\}{\renewcommand\given{\Symbol[\delimsize]}#1}
\DeclarePairedDelimiterX\conds[1][]{\renewcommand\given{\Symbol[\delimsize]}#1}
\DeclarePairedDelimiterX\condn[1]{}{}{\renewcommand\given{\Symbol[\delimsize]}#1}

\newcommand{\Exp}{\mathbb{E}}

\newcommand{\Expcr}[2][]{\Exp\condr[#1]{#2}}

\newcommand{\dif}{\:\mathrm{d}}
\let\Set\condc
\DeclareMathOperator{\Cov}{Cov}

\usepackage{xspace}

\crefname{equation}{}{}
\Crefname{equation}{Equation}{Equations}
\crefname{proposition}{proposition}{propositions}
\creflabelformat{proposition}{#2#1#3}
\crefname{figure}{figure}{figures}
\creflabelformat{figure}{#2#1#3}
\crefname{table}{table}{tables}
\creflabelformat{table}{#2#1#3}
\crefname{lem}{lemma}{lemmas}
\creflabelformat{lemma}{#2#1#3}
\crefname{thm}{theorem}{theorems}
\creflabelformat{theorem}{#2#1#3}
\crefname{corollary}{corollary}{corollaries}
\creflabelformat{lemma}{#2#1#3}
\crefname{ass}{assumption}{assumptions}
\crefname{enumi}{part}{parts}
\creflabelformat{enumi}{#2#1#3}
\crefname{ex}{example}{examples}
\creflabelformat{enumi}{#2#1#3}

\newcommand{\maligned}[1]{\left\{\begin{aligned}#1\end{aligned}\right.}

\newcommand{\sV}{\mathscr{V}}
\newcommand{\sP}{\mathscr{P}}

\newcommand{\comment}[1]{}
\usepackage{marginnote}

\newcommand{\one}{\mathbb{1}}

\allowdisplaybreaks

\newcommand{\drop}[1]{\todo[inline,color=LightGrey]{stuff at least one of us thinks can be dropped has been omitted here.}}

\usepackage{xpatch}

\xpatchcmd{\proof}
{\itshape}
{\bfseries}
{}
{}
\usepackage{afterpage}

\usepackage{tikzsymbols}

\newcommand{\lh}{\R{\hat\ell}}
\newcommand{\Rh}{\R{\hat R}}
\newcommand{\Bh}{\R{\hat B}}
\newcommand{\bt}{\tilde{\beta}}
\newcommand{\ah}{\R{\hat a}}

\makeatletter
\let\amstexbig\big
\def\newbig#1{%
	\ifx#1|%
	\expandafter\@firstoftwo
	\else
	\expandafter\@secondoftwo
	\fi
	{\big@bar}%
	{\amstexbig{#1}}%
}
\AtBeginDocument{\let\big\newbig}
\def\big@bar{\bBigg@{1.1}|}
\makeatother

\newcommand{\fhat}{\R{\hat f}}
\newcommand{\Fhat}{\R{\hat F}}
\newcommand{\Lhat}{\R{\hat L}}
\newcommand{\dsum}{\displaystyle\sum}
\newcommand{\dint}{\displaystyle\int}

\newcommand{\Nhat}{\R{\hat N}}
\newcommand{\Dhat}{\R{\hat D}}

\newcommand{\ignore}[1]{}

\newcommand{\phat}{\R {\hat p}}

\begin{document}
\title{Estimates of derivatives of (log) densities and related objects}
\author{Joris Pinkse and Karl Schurter\thanks{Corresponding author \url{kschurter@psu.edu}} }
\date{May 2020}

\maketitle
\begin{abstract}
We estimate the density and its derivatives using a local polynomial approximation to the logarithm of an unknown density $f$. The estimator is guaranteed to be nonnegative and achieves the same optimal rate of convergence in the interior as well as the boundary of the support of $f$. The estimator is therefore well--suited to applications in which nonnegative density estimates are required, such as in semiparametric maximum likelihood estimation. In addition, we show that our estimator compares favorably with other kernel--based methods, both in terms of asymptotic performance and computational ease. Simulation results confirm that our method can perform similarly in finite samples to these alternative methods when they are used with optimal inputs, i.e.\ an Epanechnikov kernel and optimally chosen bandwidth sequence. Further simulation evidence demonstrates that, if the researcher modifies the inputs and chooses a larger bandwidth, our approach can even improve upon these optimized alternatives, asymptotically.  We provide code in several languages.
\end{abstract} 
\clearpage


\section{Introduction}
\label{section:intro}

We propose a new nonparametric estimator for (the logarithm) of a density function and its derivatives that attains the optimal rate of convergence both in the interior and at the boundary of the support.  Our density estimator is available in closed form and is guaranteed to be positive unlike several alternatives, which is appealing in some applications and critical in others, such as in semiparametric maximum likelihood estimation \citep[see e.g.][]{klein1993efficient}.\footnote{Klein and Spady square their density estimates to ensure positivity.}  The new methodology differs from the previous literature in that it first estimates a function's derivatives, which, if desirable, can then be used to construct an estimate of the function itself.  Our general estimation strategy can also be applied to obtain estimates of other quantities of economic interest, including the density in regression discontinuity design models, the (reciprocal) of the propensity score, the inverse bid function in auction models, and any other application in which the density appears inside a logarithm or a denominator.\footnote{The object of interest here is nonparametric in nature, i.e.\ it does not necessarily get averaged out as it does in e.g.\ \citet{lewbel2007simple}.}

Specifically, we consider an i.i.d.\ sequence of random variables $\Set{\R x_1,\dots, \R x_n}$ with $\R x_i$ distributed according to some unknown distribution $F$ with density $f>0$ on its support $[0,\mathscr{U})$, where $\mathscr{U}$ can be infinite. 
  The standard Rosenblatt--Parzen (RP) kernel density estimator is inconsistent at the boundary and is typically badly biased in finite samples at values of $x$ near the boundary. In contrast, our method employs a local polynomial approximation of $L\parens{x}=\log f\parens{x}$ to obtain asymptotically normal estimates of $L$ and its derivatives, away from, at, or near the boundary.
 
An advantage of using a polynomial approximation to the log--density instead of the density is that the estimated density can be guaranteed to be positive, which is not true for alternative boundary correction methods that use boundary kernels or a local polynomial approximation of $f$ \citep{Cheng1997auto,Karunamuni1998endpoints} of $F$ \citep{Lejeune1992smooth,Cattaneo2019simple}.\footnote{An exception is \citet{jones1996simple}.} Unlike \cite{Loader1996likelihood} and \cite{Hjort1996local}, however, the computation of our estimator does not require solving a nonlinear system of equations that involves numerical integration. In fact, the estimator of derivatives of $L$ may be expressed as the solution to a linear system of (weighted) local averages. Thus, our method can be characterized as a local method of moments, similar in spirit to local likelihood density estimation \citep[and much subsequent work]{Loader1996likelihood, Hjort1996local} but computationally more similar to local polynomial regression \citep{Lejeune1992smooth, Cheng1997auto, Karunamuni1998endpoints,Cattaneo2019simple}. We therefore retain the computational ease of a local polynomial regression while eliminating the possibility of negative density estimates.  The estimator of $L,f$ itself then obtains in explicit form with estimates of the derivatives of $L$ as inputs.

Apart from its numerical advantages, our estimator for the density has the same first order asymptotic properties when applied with the same bandwidth as the local likelihood estimator. When applied with a larger bandwidth, however, our estimator achieves a smaller asymptotic mean squared error. We cannot generally compare the bias of our method with the biases of methods that use a polynomial approximation to $f$, but our estimator has the same asymptotic variance as traditional methods when they are applied using an optimal kernel and bandwidth sequence. Hence, our local polynomial approximation to $L$ can be expected to outperform alternative estimators for $f$ in finite samples when our bias is smaller, e.g.\ when the log--density is in fact polynomial.

In large enough samples, an asymptotically unbiased version of our estimator achieves a smaller variance and therefore a smaller mean square error than the optimized alternatives. Importantly, our estimator realizes this improved performance without sacrificing nonnegativity and continuity of the estimated density, as would be required to achieve the same asymptotic distribution using alternative methods.\footnote{One could replace negative density estimates produced by alternative methods by zero, but that is both clunky and would not help in cases in which the density must be positive.}

We also note that the log--density or its derivatives may be of direct interest to the researcher, in which case our method may be an attractive alternative to transforming estimates of $f$ and its derivatives to obtain the desired estimates. For instance, the generalized reflection method of \cite{Karunamuni2005boundary} and \cite{Karunamuni2008some} requires an estimate of $L'\parens{0}$ which they obtain using a finite difference approximation. Estimates of $L$ can moreover be used as an input into other objects.  One case that has already been mentioned is semiparametric maximum likelihood estimation, of which \citet{klein1993efficient} is a classical example in which the likelihood objective can be written as a function of the log--density.  But there are other important examples.  For instance, in regression discontinuity design, estimation of the density at the discontinuity point can be of interest \citep[see e.g.][]{Cattaneo2019simple}.  A second example would be the estimation of propensity scores which are of importance in the estimation of treatment effects.  A final example is that of the estimation of auction models for which a version of our estimator can be used to obtain direct estimates of the inverse strategy function; see e.g.\ \citet{hickman2015replacing,pinkse2019estimation}.  These examples and more are discussed in \cref{section:applications}.  Finally, we provide code in several languages, including Julia and R at \url{https://github.com/kschurter/logdensity}.

\newcommand{\mhx}{m_{hx}}
\newcommand{\bh}{\R{\hat\beta}}
\makeatletter
\newcommand{\vast}{\bBigg@{3}}
\newcommand{\Vast}{\bBigg@{4}}
\newcommand{\vastl}{\mathopen\vast}
\newcommand{\vastm}{\mathrel\vast}
\newcommand{\vastr}{\mathclose\vast}
\newcommand{\Vastl}{\mathopen\Vast}
\newcommand{\Vastm}{\mathrel\Vast}
\newcommand{\Vastr}{\mathclose\Vast}
\makeatother

\section{Estimator}
\label{section:estimator}

We now discuss our main estimator, postponing the discussion of applications and variants to \cref{section:applications}.
Let $L\parens{y} = \log f(y)$ denote the log--density function and assume it to possess at least $S+1\geq 2$ derivatives at $x$, the point at which we wish to estimate $L$. Our estimator will be in the kernel family of estimators and we denote our bandwidth by $h$. To specifically allow for $x$ approximating the boundary, we introduce the notation $z = \min\parens{x/h, 1}$. 

In the first step of our estimation procedure, we estimate derivatives of $L$, which are subsequently used to construct an estimate of $L$ itself.
Our estimator of derivatives of $L$ is based on the fact that for any differentiable function $g:\Reals \to \Reals^S$ with support $\sparens{-z,1}$ and for which $g\parens{-z}=g\parens{1}=0\in\Reals^{S}$, we have
\begin{multline}
\label{eq:population system}
\frac1{h^2}\int_{x-zh}^{x+h} g'\parens[\Big]{\frac{y-x}{h}} f\parens{y} \dif y = -\frac{1}{h}\int_{x-zh}^{x+h} g\parens[\Big]{\frac{y-x}{h}}L'\parens{y}f\parens{y}\dif y\\
= - \frac1h \int_{x-zh}^{x+h} g\parens[\Big]{\frac{y-x}{h}} \sum_{s = 1}^{S+1} \frac{ L^{\parens{s}}\parens{x} \parens{y-x}^{s-1} }{\parens{s-1}!}f(y)\dif y + o\parens[\big]{h^S},
\end{multline}
where $L^{\parens{s}}$ denotes the $s$--th derivative.  The above follows from integration by parts under the assumption that $f$ is bounded on the domain of integration and a Taylor expansion of $L'$ around $x$. We define $\beta_s = L^{\parens{s}}\parens{x}$ and gather these coefficients into a vector $\beta = \sparens{\beta_1 \dots \beta_S}^\tr$.   We then estimate integrals on the right and left sides of \cref{eq:population system} by their sample analogs and estimate $\beta$ by solving
\< \label{eq:sample system}
  \frac1{nh^2} \sum_{i=1}^n g'\parens[\Big]{\frac{\R x_i-x}{h}} =  -\sum_{s=1}^S \bh_s \frac1{nh}\sum_{i=1}^n g\parens[\Big]{\frac{\R x_i-x}{h}}\frac{(\R x_i - x)^{s-1}}{\parens{s-1}!}\,.
\>
Since \cref{eq:sample system} is linear in $\bh$, the solution will generally be unique. Moreover, we can choose $g$ such that our estimator $\bh$ of $\beta$ is in closed form. 

There are many functions $g$ that satisfy the desiderata outlined above.  For the purpose of providing examples, we choose $g$ to be a vector whose $j$--th element is
\[
g_j\parens{t} = \parens{t+z}^{j}\parens{1-t} \one\parens{-z \leq t \leq 1}.
\]
In \cref{section:optimalg} we describe the sense in which this choice of $g$ is in fact optimal.

\begin{ex}
If $S=1$ then $\bh_1$ is simply the derivative of the logarithm of the kernel density estimator with kernel
$6\one\parens{-z\leq t \leq 1}\cparens{z+t\parens{1-z}-t^2} / \parens{1+z}^3$, which simplifies to an Epanechnikov kernel if $z=1$.  \qed
\end{ex} 

\begin{ex}
\label{ex:relation to kernel regression}
If $z=1$ and $g_{j}(u) = k(u) u^{j-1}$ for some symmetric, nonnegative kernel function $k$ then \cref{eq:population system} represents the first--order condition for the minimizer of the least squares criterion in a local polynomial regression of $L'(x_{i})$ on $x_{i}$, which would be infeasible because $L'(x_{i})$ is not observed.
\end{ex}
\Cref{ex:relation to kernel regression} illustrates that the integration by parts in \cref{eq:population system} can be viewed as a device for obtaining a feasible set of local moment conditions from an infeasible set of moment conditions involving $L'$. Thus, $g$ fulfills a role similar to a kernel, but the restrictions we impose are different. Indeed, we require $g(-z)=g(1)=0$ so that $g(1)f(x+h)-g(-z)f(x-zh)$ is zero after integration by parts. 

Now, once we have an estimator $\bh$ of the derivatives of $L$ at $x$, we can use it to construct estimators of $f\parens{x}$ and $L\parens{x}$.  Indeed, substituting our approximation for the log--density in $\int_{x-zh}^{x+h} m_z\parens{x}f\parens{x}\dif x$ and rearranging suggests an estimator for $\beta_0$ similar to that of \citet{Loader1996likelihood}:
\< \label{eq:fhat def}
 \fhat\parens{x} = \frac{\fhat_m\parens{x}}{ \dint_{-z}^1 m_z\parens{t} \exp\parens[\Big]{\dsum_{s=1}^S \frac{\bh_s t^s h^s}{s!} } \dif t }, 
\>
where $\fhat_m$ is a RP estimator using a nonnegative kernel $m_z$ with support $\sparens{-z,1}$.  It should be apparent that $\fhat\parens{x}$ cannot be negative and is zero only if there are no data on the interval $\sparens{x-zh,x+h}$.

If $z<1$ then $m_z$ can be thought of as a traditional boundary kernel\footnote{$k / \int_{\max\parens{x/h-1,0}} k$; see e.g.\ \citet{gasser1979kernel}.} such that the bias in the numerator of \cref{eq:fhat def} is $O\parens{h}$.  The role of the denominator is that it is an (asymptotically) biased estimator of the number one.  Indeed, the denominator bias compensates for the numerator bias such that for $S=1$, the bias of $\fhat\parens{x}$ is again $O\parens{h^2}$.  
We define $\Lhat\parens{x} =\log \fhat\parens{x}$ and show that its asymptotic bias is $\beta_{S+1} h^{S+1}$ times a constant that is independent of both $n$ and fixed $x$.\footnote{For the case $x=zh$, the asymptotic bias does depend on $z$.}

Computing $\fhat$ is relatively simple because $\bh$ is simply a local least squares statistic and $\fhat$ is a ratio.  Although $\fhat$'s denominator contains an integral, for values of $S\leq 2$, which will be the most common scenario, the denominator in \cref{eq:fhat def} obtains in closed form if $m_z$ is a truncated Epanechnikov; for $S=1$ this is demonstrated in \cref{ex:denom} below.\footnote{For $S=2$ and $\bh_1<0$ it would however involve a normal distribution function and for $\bh_1>0$ a similar integral.}  For other kernels and greater values of $S$, an asymptotically equivalent closed form expression can be obtained by expanding the denominator in \cref{eq:fhat def} in terms of exponential Bell polynomials \citep{bell1927partition}.\footnote{For $S=1$ we
get that that the exponential in the denominator in \cref{eq:fhat def} can be written as
$1+\bh_1 th + \bh_1^2 t^2 h^2 + o_p\parens{h^2}$.  For $S=2$ we get
$1+\bh_1 th + \parens{2\bh_2 +\bh_1^2}t^2h^2 + \parens{3\bh_1\bh_2+\bh_1^3}t^3h^3 + o_p\parens{h^3}$.} Standard numerical integration methods can also be applied to this integral, in which case an advantage of our method is that this integral only needs to be computed once rather than at each iterate of a maximization routine, as in a local maximum likelihood approach.

The following examples compare the asymptotic behavior of our estimator and traditional approaches to kernel density estimation at and away from the boundary.  \Cref{ex:denom} obtains a closed form expression for the denominator in \cref{eq:fhat def}, which is used in \cref{ex:Epanechnikov away from the boundary S=1,ex:Epanechnikov at the boundary S=1} to obtain explicit expressions for the special cases $z=1$ and $z=0$ (away from the boundary and at the boundary, respectively).
\begin{ex}\label{ex:denom}
Let $\nu=\nu\parens{z} = 3 /\parens{2+3z-z^3} = 3 / \cparens{\parens{1+z}^2\parens{2-z}}$ and suppose that $m_z$ is a truncated Epanechnikov, i.e.\
$m_z\parens{t} = \nu \parens{1-t^2} \one\parens{-z\leq t\leq 1}$.  If $S=1$ then the denominator in \cref{eq:fhat def} is $\chi_1\parens{\bh_1 h}$, where
\[
\chi_1\parens{t} =
\begin{cases}
1, & t = 0, \\
\nu t^{-3}
\sparens[\big]{
\parens{2t - 2} \e^t
-
\cparens[\big]{-2-2tz + t^2 \parens{1-z^2} } \e^{-zt}
 },  & t \neq 0.
\end{cases} \tag*{\raisebox{-4ex}{\qed}}
\]
\end{ex}

\begin{ex} \label{ex:Epanechnikov away from the boundary S=1}
Suppose $S=1$.  If $x\geq h$ then $g\parens{t} = - \parens{1-t^2} \one\parens{\abs{t}\leq 1}$, which is proportional to minus the Epanechnikov kernel.  So then $\bh_1$ is simply the derivative of the logarithm of the RP estimator using the Epanechnikov kernel.  The denominator in \cref{eq:fhat def} is then exactly
\[
\begin{cases}
 \dfrac3{2\bh_1^3h^3} \cparens[\big]{ \exp\parens{\bh_1 h} \parens{\bh_1h-1} + \exp\parens{-\bh_1 h} \parens{\bh_1h+1} }, & \bh_1 \neq 0,\\
 1, & \bh_1 =0.
 \end{cases}
\]
The denominator can be expanded around $h=0$ to obtain the approximation
\(
 \Lhat\parens{x} \approx \log\fhat_m\parens{x} - \log \parens[\big]{1+ \bh_1^2 h^2/10}.
\)
It is well--known that the bias of $\log\fhat_m\parens{x}$ is $h^2 f''\parens{x} \mathrel/ 5 f\parens{x} + o\parens{h^2}$.  The bias of $\Lhat\parens{x}$ is by the mean value theorem then seen to be
$h^2 \beta_2 / 5 + o\parens{h^2}$.  Thus, the bias we introduced in the denominator offsets the bias present in the numerator.
  \qed	
\end{ex}

In \cref{ex:Epanechnikov away from the boundary S=1} we took $x>h$ and hence $z=1$ to provide intuition.  In the following example we consider what happens at the boundary, i.e.\ if $x=z=0$.

\begin{ex} \label{ex:Epanechnikov at the boundary S=1}
Again suppose that $S=1$, but now let $x=z=0$.  Then $g\parens{t} = - \parens{1-t}t$, such that $\bh_1$ is now the derivative of the kernel density estimator at zero using the kernel $6\parens{1-t}t\one\parens{t\leq 1}$.  

If we again use an Epanechnikov in \cref{eq:fhat def} then now the denominator becomes for $\bh_1 \neq 0$,
\[
\frac32 \frac{2 + \bh_1^2 h^2 - \exp\parens{\bh_1 h}\parens{2-2\bh_1h}}{\bh_1^3h^3}
=
  1 + \frac{3\bh_1h}8 + \frac{\bh_1^2h^2}{10} + o\parens{h^2}.
\]
Our results show that the bias of $\bh_1$ is $\beta_2 h/2 + o\parens{h}$, such that the denominator bias is $3\beta_1 h/8 + 3\beta_2h^2/16 + \beta_1^2 h^2/10 + o \parens{h^2}$.
Further,
\[
  \Exp\fhat_m\parens{0} 
 =
 f\parens{0} + \frac{3hf'\parens{0}}8 + \frac{f''\parens{0}h^2}{10} + o\parens{h^2},
\]
such that the bias of $\fhat\parens{0}$ is now
$-7h^2 \beta_2 f\parens{x}/ 80 + o\parens{h^2}$.  The bias of $\Lhat\parens{0}$ is by the delta method hence $-7h^2\beta_2 /80+o\parens{h^2}$.  Again, the bias we introduced in the denominator offsets the bias in the numerator.  \qed
\end{ex}

\Cref{ex:Epanechnikov at the boundary S=1} demonstrates that, unlike traditional boundary kernel estimators, the bias of our estimator is $O\parens{h^2}$ at the boundary, also.  It may seem odd that the bias in \cref{ex:Epanechnikov at the boundary S=1} is less than that in \cref{ex:Epanechnikov away from the boundary S=1} but note that the variance will be larger at the boundary and one would hence generally choose a greater bandwidth.

\section{Limit results}

\subsection{Derivatives of $L$}

We first derive limit results for the vector $\bh$ of estimates of derivatives of $L$.  Since our estimator $\bh$ is defined as the inverse of a matrix times a vector, its bias is the inverse of a matrix times a vector, also.

To simplify expressions for the asymptotic bias and variance of our estimator, we introduce the following objects which depend only on the choice of function $g$ (which in turn depends on the proximity to the boundary $z$), as well as a diagonal matrix that depends on $h$. Let
\<\label{eq:Omega}
\Omega_{js}=\Omega\parens{z} = \frac1{\parens{s-1}!}\int_{-z}^{1}-g_{j}\parens{t} t^{s-1}\dif t
= \sum_{t=0}^{s-1}  \frac{\parens{-1}^{s-t+1}\parens{s-t} j!}{\parens{j+s+1-t}! t!}  \parens{1+z}^{j+s+1-t}, \quad 
j,s=1,2,\dots,
\>
and let $\Omega,b,\Lambda$ be defined as
	\[
\Omega =	\begin{bmatrix}
	\Omega_{11} & \Omega_{12} & \cdots & \Omega_{1S} \\
	\Omega_{21} & \ddots & \ddots & \vdots \\
	\vdots & \ddots & \ddots & \Omega_{S-1,S} \\
	\Omega_{S1} & \cdots & \Omega_{S,S-1} & \Omega_{SS}
\end{bmatrix},
\qquad
b = \Omega^{-1}
\begin{bmatrix}
	\Omega_{1,S+1} \\
	\Omega_{2,S+1} \\
	\vdots \\
	\Omega_{S,S+1}
\end{bmatrix}, 
\qquad
\Lambda = \begin{bmatrix} 1 & & \\ & \ddots & \\ & & h^{S-1} \end{bmatrix}.
\]
Let further $V \in \Reals^{S\times S}$ have $\parens{j,s}$ element equal to 
\[
 \int_{-z}^1 g_j'\parens{t}g_{s}'\parens{t}\dif t = \frac{2  js \parens{1+z}^{j+s+1}}{\parens{j+s+1}\parens{j+s}\parens{j+s-1}}.
\]
We are now in a position to state our first theorem.
\begin{thm} \label{thm:beta}
Assume $L$ is $S+1$ times continuously differentiable in a neighborhood of $x$.  Let $h \to 0$ and $n h^3 \to \infty$ as $n\to\infty$.  For a vector $\tilde\beta$ defined in \cref{app:definitions},
\< \tag*{\qed}
  \Lambda \parens{\tilde\beta - \beta} = h^S\beta_{S+1}  b + o_p\parens{h^S},
  \qquad
 \sqrt{nh^3} \Lambda\parens{\bh - \tilde\beta} \convd
 N\parens[\Big]{
 0
 ,\:
 \parens{\Omega^\tr V^{-1} \Omega}^{-1} \bigm/ f\parens{x}
}.
\>	
\end{thm}
The ``in a neighborhood'' condition comes from the fact that we specifically allow $x=zh$.

\begin{ex}
 For $S=1$ the bias and variance expressions of $\bh_1$ simplify to 
 \(
\beta_2 h \parens{1-z}/2
 \)
 and
\(
 12 / \cparens{ f\parens{x} \parens{1+z}^3},
\)	
respectively. The interior case ($z=1$) is more favorable than the boundary case ($z=0$), as expected. \qed
\end{ex}

\begin{ex}
For $S=2$ the bias and variance expressions for $\bh_1$ are
\(
 -\beta_3 h^2 \parens{ 1-3z+z^2} / 10
\)
and
\(
 48 \parens{4-7z+4z^2}/\cparens{f\parens{x}\parens{1+z}^5},
\)
which is again more favorable in the interior than at the boundary. \qed
\end{ex}	
	
\subsection{Density}

We now continue with the results for $\fhat\parens{x}$.

Let  $c_{msz} = \int_{-z}^1 m_z\parens{t} t^s \dif t \mathrel/ s!$ and let $c_{mz}$ be a vector with elements $c_{m1z},\dots,c_{mSz}$.   Let further
\(
\Omega_z\parens{t} = m_z\parens{t} - c_{mz}^\tr \Omega^{-1} g'\parens{t},
\)
and define
\(
\sV = f\parens{x}\int_{-z}^1 \Omega_z^2\parens{t} \dif t
\)
 and 
 \(
 \sB =  f\parens{x}\beta_{S+1}\Xi_{f}\int_{-z}^{1}\omega_{z}\parens{t}t^{S+1}\dif t \mathrel/\parens{S+1}!= f\parens{x}\beta_{S+1} \Xi_f \parens[\big]{ c_{m,S+1,z} -  c_{mz}^\tr b} ,
\)
for $\Xi_f$ a constant defined in the statement of \cref{thm:f}. Because $m_{z}$ integrates to one and $\int_{-z}^{1} g'\parens{t}\dif t = 0$ and $\int_{-z}^{1}\Omega^{-1}g'(t)t^{s}\dif t \mathrel/ s! =\Omega^{-1}\Omega_{\cdot s}$ is the $s$--th standard basis vector in $\Reals^{S}$, the function $\omega_{z}$ is a kernel of order $S+1$ or higher. To be clear, $\omega_{z}$ is not used to compute the density estimate; rather, it is a convenient object that arises in the asymptotic theory.

\begin{thm} \label{thm:f} 
	Assume $L$ is $S+1$ times continuously differentiable in a neighborhood of $x$, that $f\parens{x}>0$, and that $0 
	\leq \Xi_f^2 = \lim_{n\to\infty} nh^{2S+3}<\infty$.  Then
\( 
 \sqrt{nh} \cparens{\fhat\parens{x} - f\parens{x}} 
 \convd
 N\parens{\sB,\sV}.
\) \qed	
\end{thm}

The asymptotic bias of our estimator is zero in some instances. For example, if $S=2$ and $z=1$ then the asymptotic bias is zero whenever $m$ is a symmetric kernel function; this is natural since this is effectively equivalent to choosing a higher order kernel $\omega_{z}$, albeit that unlike higher order kernel density estimates, our estimates cannot be negative.

The following two examples derive the $\Omega_z$ functions for the case in which both $m_z$ is a uniform and $x$ is at the boundary and the case in which $m_z$ is a truncated Epanechnikov and $x$ is anywhere.

\begin{ex}
\label{ex: linear vs quadratic uniform}
Suppose that $m_z$ is a uniform and $x=z=0$.
If $S=1$ then $\Omega_0\parens{t} = \parens{4-6t} \one \parens{0\leq t\leq 1}$ and $\omega_1\parens{t} = \frac{1}{2}\one\parens{|t|\leq 1}$.
If instead $S=2$ then $\Omega_0\parens{t} = \parens{9-36t+30t^2} \one\parens{0\leq t\leq 1}$ and $\omega_{1}\parens{t} = \frac{3}{8}\parens{3-5 t^{2}}\one\parens{|t|\leq 1}$. \qed
\end{ex}

{
\begin{ex}
	If $m_z$ is a truncated Epanechnikov and $S=1$ then $m_z\parens{t} = \nu\parens{1-t^2} \one\parens{-z\leq t\leq 1}$ with
	$\nu= 3 / \cparens{ \parens{1+z}^2 \parens{2-z}}$, such that
	$c_{m1z} = \nu \parens{1-z^2}^2 /4 = 3 \parens{1-z}^2 / \cparens{4\parens{2-z}}$,
	$c_{m2z} = \parens{2-4z+6z^2-3z^3} / \cparens{ 10 \parens{2-z}}$,
	and $b = \parens{1-z}/2$, which produces
	\[
	\sB
	=
	\beta_2 \Xi f\parens{x}\frac{3z^3+29z-7-21z^2}{40\parens{2-z}},
	\]
	where the ratio equals $1/10$ for $z=1$ and $-7/80$ for $z=0$.   To get $\sV_f$ note that
	\(
	\Omega_z\parens{t} =  \nu \cparens{ 2\parens{1+z} \parens{1-t^2} - 6 \parens{1-z}^2t + 3 \parens{1-z}^3} \mathbin/ \cparens{{2\parens{1+z}}},
	\)
	which produces
	\[
	\sV = f\parens{x} \frac{108-180 \parens{1+z} + 120 \parens{1+z}^2 - 36 \parens{1+z}^3 + 4.05 \parens{1+z}^4}{\parens{1+z}^3\parens{2-z}^2},
	\]
	which equals $0.6 f\parens{x}$ for $z=1$ and $4.01 f\parens{x}$ for $z=0$.  
	\qed
\end{ex}
}

As the above two examples demonstrate, deriving the asymptotic bias and variance for generic $z$ can be a messy but straightforward exercise.
\section{Asymptotic comparisons}
\label{section:optimalg}
In this section, we explore the optimal $(g,m_{z})$ in the local linear case $S=1$ and compare our optimized estimator with existing methods. We show that the above choice of $g$ and $m_{z}$ achieve the same asymptotic variance as an optimal RP estimator in the interior ($z=1$), while their respective biases cannot be compared in general. We then consider the optimal choice of $g$ and $m_{z}$ at the boundary ($z=0$), where we show that the truncated Epanechnikov $m_{z}$ paired with $g(t) = (t+z)(t-1)^{2}$ attains the same variance as an optimal boundary kernel \citep{Karunamuni1998endpoints}, though the biases are again incomparable because our estimator's bias is a function of $f(x)L''(x)$ rather than $f''(x)$ as in the case of RP estimators. We are, however, able to compare the asymptotic performance of our estimation method with a local--likelihood based estimator. We show that our method with the cubic choice $g\parens{t}=\parens{t+z}\parens{1-t}^{2}$ attains the same asymptotic mean squared error (AMSE) in the interior and is more efficient at the boundary than the estimator in \citet{Loader1996likelihood} with an Epanechnikov kernel.

\subsection{Optimal choice of $g$ and $m_{z}$ in the local linear case}
Letting $\chi^{5} = \lim_{n\to\infty} h^{5}nf\parens{x}L''\parens{x}^{2}$, the AMSE of $\hat f\parens{x}$ only depends on $h$ and $\parens{g,m}$ through a multiplicative constant that can be written in terms of $\chi$ and the second--order kernel $\omega_{z}$:
\<
\label{eq:amse}
\sparens[\bigg]{\chi^{4} \parens[\bigg]{\int_{-z}^{1}\omega_{z}\parens{t}t^{2}/2 \dif t}^{2}+ \chi^{-1} \int_{-z}^{1}\omega_{z}\parens{t}^{2}\dif t}\sparens[\bigg]{f\parens{x}^{6/5}L''\parens{x}^{2/5}n^{-4/5}}\,.
\>
Unlike the typical approach to comparing kernels in kernel density estimation, in which one considers the optimal choice of $\chi$ as a function of $\omega_{z}$, we treat $\chi$ as fixed and seek to minimize the asymptotic MSE over $\omega_{z}$ instead of $(\omega_{z},\chi)$. We do so for two reasons. First, for values of $x$ near but not at the boundary, the function $\omega_{z}$ depends on $h$ through $z$, with the result that the first--order condition for optimality of the bandwidth is generally insufficient for the global minimum of the AMSE as a function of $h$. Second, many combinations of $g$ and $m_{z}$ yield a function $\omega_{z}$ that achieves zero asymptotic bias, which implies that there does not exist a finite optimal $\chi$.\footnote{One could assume an additional derivative of $f$, in which case the optimal bandwidth sequence would be proportional to $n^{-1/7}$ and one might also consider using a quadratic approximation ($S=2$).} 

In light of the apparent similarity between \cref{eq:amse} and the corresponding expression for the AMSE of RP estimators, one might expect $\omega_{1}=3(1-t^{2})/4$ to be optimal using our method for the same reason that the Epanechnikov kernel is an optimal second--order kernel for use in RP estimation. Although we will eventually recommend $\omega_{1}=3\parens{1-t^{2}}/4$ for a particular value of $\chi$, our reasoning is different in two important ways. First, we do not require $\omega_{z}\geq 0$ as in \cite{Epanechnikov1969}, because this restriction is not necessary to guarantee positive density estimates. Nor do we require $\omega_{z}\parens{1} = 0$ and $\omega_{1}\parens{-1}=0$ as in \cite{Muller1984} because these restrictions are not needed to ensure the density estimate is continuous in $x$. We place these restrictions on $m_{z}$, instead. Second, these constraints on $m_{z}$ and the maintained assumptions on $g$ are not binding if one minimizes the AMSE over $\parens{\omega_{z},\chi}$; for instance, $m_{1}\parens{t} = 3\parens{1-t}^{2}\parens{1+t}/4$ and $g\parens{t} = \parens{t+z}\parens{1-t}\parens{t^{2} + 2t - 1}$ yields $\omega_{1}\parens{t}=3\parens{3-5 t}^{2}/8$, which is the fourth--order kernel that minimizes the variance conditional on achieving zero bias and minimizes the AMSE as $\chi$ tends to infinity. Hence, when translated into the context of our estimator, the typical constraints on second--order kernels in RP estimation do not yield an interior solution to the optimal choice of $\parens{\omega_{z},\chi}$.

Thus, we seek a pair $\parens{g, m_{z}}$ with $m_{z}\geq 0$ that yields the optimal $\omega_{z}$ given a particular $\chi$, though we do allow $\chi$ to depend on $z$ so that the bandwidth sequence may be larger at the boundary than in the interior. The necessary conditions for the minimizer of the AMSE in \cref{eq:amse} imply that the optimal $\omega_{z}$ is quadratic, which generally implies a quadratic $m_{z}$ and a cubic $g$. The minimizer is not unique, however, because many pairs will produce the same $\omega_{z}$ and therefore the same asymptotic distribution. In fact, if the ``first moment'' of $m_{z}$ is zero, the AMSE in $\hat f$ does not depend on the choice of $g$ because $c_{m1}=0$ and $\omega_{z}=m_{z}$. A case such as this arises, for example, when $\chi = 15^{1/5}$, $z=1$, and one minimizes the AMSE by letting $m_{z}$ be the Epanechnikov kernel.

We focus on the choice of $\chi=15^{1/5}$ for two reasons. First, the constants in the expressions for the limiting bias and variance of $\hat f$ are the same as those found in the asymptotic bias and variance of the RP estimator for $f$ using the Epanechnikov kernel. This choice of $\chi$ and $m_{z}$ therefore provides a benchmark for comparison with an optimal RP estimator, because, for a fixed bandwidth $h$, the magnitude of our bias to the RP density estimator's bias depends on the ratio of $f(x)L''(x)$ to $f''(x)$ but our asymptotic variances are the same.  And, second, the higher--order bias terms may be non-negligible when $\chi$ is too large, which could worsen the asymptotic approximation to the MSE in finite samples. Lacking a useful definition of ``too large,'' we default to a familiar choice.

Though we treat $\chi=15^{1/5}$ as fixed, this value is the optimal bandwidth scaling factor to use with the Epanechnikov kernel. We note, however, that this does not imply that the Epanechnikov $m_{z}$ and $\chi=15^{1/5}$ attain the minimum over all pairs $(m_{z},\chi)$. One can achieve a smaller AMSE at $z=1$ given a larger bandwidth by using a cubic $m_{z}$ and quartic $g$.\footnote{No symmetric nonnegative $m_{1}$ can improve on the AMSE of the Epanechnikov kernel because $g$ does not affect the limiting distribution. An asymmetric kernel---e.g.\ cubic $m_{1}$ with $m_{1}\parens{-1}=m_{1}\parens{1}=0$---and carefully selected $g$ is needed in order to improve on the AMSE at $z=1$.} Indeed, the above choice of $\omega_{1}(t) = 3\parens{3-5t^{2}}/8$ is an extreme example in which the asymptotic bias is zero. Its AMSE is smaller whenever $\chi > 3\times 15^{1/5}/2$, and its asymptotic variance is smaller as long as $\chi > 15^{6/5}\mathrel/8$, i.e.\ 1.875 times larger than the bandwidth used with the Epanechnikov kernel. Thus, even though our estimator's bias is generally not comparable to the bias of estimators that employ a local constant (RP) or local polynomial \citep{Lejeune1992smooth,Cattaneo2019simple} approximation to $f$ or $F$, the asymptotically unbiased version of our estimator always has a smaller AMSE than these alternatives if the researcher is willing to use a large enough bandwidth.

Since the above criterion does not inform the optimal choice of $g$ for $z=1$ and $\chi=15^{1/5}$, one can choose $g\parens{t} = 1-t^{2}$ to minimize the AMSE in $\beta_{1}$. 

At the boundary ($z=0$), it is perhaps reasonable to use a bandwidth that is twice as large as that used at $z=1$, i.e. $\chi = 2 \times 15^{1/5}$. In this case, the optimal AMSE is attained with the truncated Epanechnikov and $g(t) = t (1-t)^{2}$. Interestingly, this choice implies that $\omega_{0}\parens{t} = 6\parens{1-2t}\parens{1-t}$, which is the optimal boundary kernel derived by \citet{Karunamuni1998endpoints}. Indeed, the constants in our bias and variance expressions are the same as the kernel--related constants in the limiting distribution of the RP estimator for $f\parens{0}$ given by $\sum_{i=1}^{n}\omega_{0}\parens{x_{i}/h}/\parens{nh}$, indicating that the asymptotic variances of the two estimators are the same for a fixed bandwidth. In contrast to this estimator, however, our proposed estimator is always positive.

As in the interior case ($z=1$), this choice of $m_{z}$ and $g$ is not optimal over all triples $\parens{g,m_{z},\chi}$. One can obtain zero asymptotic bias and a smaller asymptotic variance using the truncated Epanechnikov $m_{z}$, $g(t) = \parens{t+z}\parens{t-1}\parens{t-5/7}$, and any finite $\chi>15^{6/5}/4$.\footnote{These inputs yield the third--order kernel boundary $\omega_{0}\parens{t}= 9 - 36 t + 30 t^{2}$ that minimizes $\int_{0}^{1} \omega_{0}\parens{t}^{2}\dif t$} But we caution that this relatively large bandwidth---at least 3.75 times larger than the optimal bandwidth in the interior---can limit the usefulness of our asymptotic approximation to the bias in finite samples.

Finally, we note that the asymptotically unbiased $\omega_{0}\parens{t}$ and $\omega_{1}\parens{t}$ are the same as those derived for the $S=2$ case in \cref{ex: linear vs quadratic uniform}. Hence, the asymptotically unbiased local linear estimator has the same limiting distribution as an undersmoothed local quadratic estimator, i.e.\ a local quadratic estimator using a bandwidth sequence of order $n^{-1/5}$ instead of $n^{-1/7}$, though they are not numerically equivalent.

\subsection{Relative AMSE}
The AMSE of the local--likelihood estimator for $f\parens{x}$ in \citet{Loader1996likelihood} can be written in a form similar to \cref{eq:amse}. The relative AMSE of our proposed estimator and the local likelihood estimator using an optimal kernel is then given by the ratio of the multiplicative constants that scale $f\parens{x}^{6/5}L''\parens{x}^{2/5}n^{4/5}$. At $z=1$, if one uses an optimal bandwidth sequence and the Epanechnikov kernel with the local--likelihood based estimator and $\chi=15^{1/5}$ with our optimal estimator, the relative AMSE of our estimators is one. In fact, the limiting distributions are identical. At $z=0$, the optimal kernel to use with the local likelihood estimator is triangular, i.e. $k\parens{t}=\parens{1-|t|}\mathbb1\cparens{|t|\geq 1}$, which yields the same asymptotic bias and variance as our estimator using $\chi = 2\times15^{1/5}$, the truncated Epanechnikov, and $g\parens{t}=\parens{t+z}\parens{1-t}^{2}$.

We should expect our estimator with $g\parens{t}=\parens{t+z}\parens{1-t}^{2}$ and the local--likelihood estimator to perform similarly in finite samples. Thus, our density estimator's computational ease is its more salient advantage over the local--likelihood based approach, given these inputs. Of course, one could make an alternative choice of $g$ and increase the bandwidth to widen the gap in AMSE at the possible expense of a larger finite sample bias. Greater reductions in the AMSE require larger bandwidths and risk worse finite sample performance.

\subsection{Optimal inputs with polynomial approximations of higher order}
For $S>1$, the optimal $\parens{g,m_{z}}$ can again be reformulated as the optimal choice of a higher order kernel $\omega_{z}$. Unlike in RP estimation, however, the higher order kernel does not necessarily entail the possibility of negative density estimates, since one can achieve a higher order $\omega_{z}$ kernel using a nonnegative $m_{z}$ and a suitable $g$. Moreover, as in the linear case, the restrictions $\omega_{z}\parens{1}=0$ and $\omega_{1}\parens{-1} = 0$ are not necessary in order for the estimated density to be continuous. Without these restrictions we do not obtain an interior solution to the optimal combination of bandwidth and $\omega_{z}$. We would therefore fix the bandwidth when we optimize the AMSE over $\parens{g,m_{z}}$ in the higher order case, as well.


\section{Applications}
\label{section:applications}

\subsection{Treatment effects}

It is well--known \citep[see e.g.][]{hirano2003efficient} that under an unconfoundedness assumption the average treatment effect can be expressed as
\[
\Exp\parens[\Big]{\frac{ \R y \R t}{p\parens{\R x}} - \frac{ \R y \parens{1-\R t}}{1-p\parens{\R x}}},
\]
where $p$ is the propensity score, $\R y$ the outcome variable, $\R x$ a vector of regressors, and $\R t$ a binary treatment variable.  Let $p_1= \Exp \R t$ be the unconditional treatment probability.  Let further $f_1$ denote the regressor density function conditional on treatment and $f_0$ the density conditional on nontreatment.  Then $f\parens{x} = f_1\parens{x} p_1 + f_0\parens{x} \parens{1-p_1}$, which produces
\[
\frac1{p\parens{x}} = 1 + \frac{f_0\parens{x}}{f_1\parens{x}} \frac{1-p_1}{p_1}, 
\qquad
\frac1{1-p\parens{x}} = 1+ \frac{f_1\parens{x}}{f_0\parens{x}} \frac{p_1}{1-p_1},
\]
such that the reciprocals of the propensity scores only depend on the ratios of the densities and the unconditional choice probabilities. In practice, $p\parens{x}$ is often estimated using a logistic functional form and possibly a series expansion in $x$,\footnote{This casual observation is supported by the fact that the built--in propensity score matching estimator in Stata defaults to the logit model.}  which implies the logarithm of the odds ratio is a polynomial in $x$. Our log--polynomial approximation to $f_{1}$ and $f_{0}$ similarly imply the odds ratio is log--polynomial. The data $\R x$ will not typically be scalar--valued, so one could for instance use a linear index of regressors instead of the regressors themselves; see \cref{sec:ks} for an example of how one might estimate $p\parens{x}= p^*\parens{x^\tr\theta_0}$. In any case, our approach is a natural local extension to the logit series estimator for $p\parens{x}$.

\subsection{Auctions}

In first--price, sealed--bid procurement models with independent private values, it is well--known \citep[see e.g.][]{Guerre2000} that the inverse bid function is of the form $b - \bar F\parens{b} / f\parens{b}$, where $\bar F,f$ are the survivor and density functions of the minimum rival bid.  Since the support of the bid distribution is assumed to have a lower bound in this literature (costs cannot be less than zero and hence neither are bids), boundary issues are a serious concern.

So let $Ψ\parens{y} =  \bar F\parens{y} / f\parens{y}$ be the object of estimation.    One way of estimating $Ψ$ is to estimate $\bar F,f$ separately where $f$ is estimated using the machinery in the main part of this paper and $\bar F$ is estimated by the empirical survivor function.  This estimator has all the features of the estimator discussed earlier in the paper.  In particular, if the underlying cost distribution is approximately an exponential then so is the bid distribution and our estimator could be expected to work especially well.

The above approach is not specific to auction models.  Indeed, consider the hazard function $H\parens{x} = f\parens{x} / \bar F\parens{x}$. $H$ can also be estimated using the machinery developed in our paper.

\subsection{Semiparametric maximum likelihood}
\label{sec:ks}
There are many examples of semiparametric maximum likelihood estimators.  Here, we only consider a classical ones, namely the \citet{klein1993efficient} estimator of the coefficients in a semiparametric binary response model, which maximizes
\[
\sum_{i=1}^n \sparens[\big]{ \R y_i \log \phat\parens{\R x_i^\tr \theta} + \parens{1-\R y_i} \log \cparens{ 1- \phat\parens{\R x_i^\tr \theta}}},
\]
where $\phat$ is an estimator of the choice probability.  Klein and Spady apply techniques to ensure that the estimates $\phat$ are positive and less than one, including trimming and adding a sample--size--dependent constant.  Our method could be helpful since 
\(
p\parens{t} = \Pr\condr{\R y_1=1 \given \R x_1^\tr\theta=t} = f_1\parens{t} p_1 / f\parens{t},
\)
where $f_j$ is the density of the linear index for observations with $\R y_i=j$ and $p_1$ is the unconditional choice probability.  Since $f\parens{t} = p_1 f_1\parens{t} + \parens{1-p_0} f_0\parens{t}$, the infeasible contribution to the loglikelihood could be written as
\[
\R y_i \log \frac{f_1\parens{\R x_i^\tr \theta}}{f_0\parens{\R x_i^\tr\theta}} 
-
\log \parens[\bigg]{ p_1\frac{f_1\parens{\R x_i^\tr\theta}}{f_0\parens{\R x_i^\tr\theta}} + \parens{1-p_1}} + \text{constant},
\]
such that it is only the ratio of $f_1/f_0$ that matters.  Obtaining conditions under which our estimator obtains the semiparametric efficiency bound, like the Klein and Spady estimator does, are well beyond the scope of this paper.

\subsection{Regression discontinuity design}

One context in which the behavior of estimates near or at the boundary is of special importance is that of regression discontinuity design. For instance, \citet{Cattaneo2019simple} provide a test of continuity of the density function at the boundary using a boundary density estimator that is similar to the estimator in \citet{Lejeune1992smooth} in that it is based on a quadratic expansion of the distribution function.  Compared to that approach, our method requires that the density be nonzero at the boundary, which is a requirement for the regression discontinuity framework in any case.  The bottom line is that our method will work better if the log density is approximately a low order polynomial near the boundary and theirs if the density itself is approximately a low order polynomial.  This is borne out by our simulation results.

\subsection{Other boundary--correction methods}

The boundary correction method of \citet{Karunamuni2008some} requires a well--behaved estimate of $L'\parens{0}$, which is exactly what our method provides.

\section{Simulations}
\label{section:sims}

The following simulation exercise compares the performance of our estimator for the density and its derivatives near the boundary with alternatives that also employ local polynomial approximations \citep{Lejeune1992smooth, Cattaneo2019simple,Loader1996likelihood} and the generalized reflection method, which also estimates the derivative of the log--density near the boundary to remove the boundary effects of the RP estimator \citep{Karunamuni2005boundary,Karunamuni2008some}. For the local polynomial estimators, we use a local linear approximation to the density or log--density, depending on the method.\footnote{This corresponds to a local--quadratic polynomial approximation to the distribution function in \cite{Lejeune1992smooth} and \cite{Cattaneo2019simple}.}

We simulate 2000 i.i.d.\ samples of size $n=500$ and estimate $f$ at points within two bandwidths of the boundary in order to compare the estimators away from, near, and at the boundary.
The random variables are drawn from each of four parametric distributions whose densities exhibit varying behaviors near their left boundary $x=0$. The first is a beta distribution rescaled to take support on $[0,5]$ (so that right boundary is sufficiently far away from zero) with density $f_{1}\parens{x} = \theta \parens{1-x/5}^{\theta-1}/5$, which is in fact polynomial in $x$ for integer values of $\theta$, which might favor CJM, although that is not reflected in our simulations if  $\theta>3$. The second design is a normal distribution with a mean of $\theta/2$ and variance of one, truncated at zero. This density is log--quadratic, which should favor our method and Loader's. The third and fourth designs are $f_{3}\parens{x} = \parens{e^{-x} + \theta x e^{-x}}/\parens{1+\theta}$ and $f_{4}\parens{x} = \parens{e^{-x} + \theta x^{2}e^{-x}}/\parens{1+2\theta}$.

For each simulation design, we estimate $f$ and its derivative using our approach with $g_{j}\parens{t}=\parens{t+z}^{j}\parens{t-1}$, $g_{j}\parens{t}=\parens{t+z}^{j}\parens{1-t}^{2}$, and $g_{j}\parens{t}=\parens{t+z}^{j}\parens{t-1}\parens{t-5/7}$ (PS$_{1}$, PS$_{2}$, PS$_{3}$), Loader's local likelihood estimator (Loader), a local polynomial regression of the empirical CDF (LS--CJM), and a generalized reflection estimator (KZ). Wherever a kernel is required, we use the Epanechnikov kernel $k(u) = 3 \parens{1- u^{2}} /4$ or a truncated version thereof. This choice is not optimal at the boundary for Loader's estimator, but the efficiency loss is quite small.\footnote{The relative efficiency of Loader's estimator using the triangular and Epanechnikov kernels at the boundary is about 1.008, meaning the Epanechnikov kernel requires a sample size 1.008 times larger to achieve the same MSE as the triangular kernel. We therefore expect the Epanechnikov kernel with $n=504$ to have the same MSE as the triangular kernel with $n=500$.}

Where possible, we use the asymptotically optimal bandwidth sequences for $z=1$ and $z=0$. For intermediate values of $z$, we linearly interpolate the bandwidth.\footnote{Specifically, we use a bandwidth $h = h_{0}\parens[\big]{1-\min\cparens[\big]{\frac{x}{h_{1}},1}} + h_{1}\min\cparens[\big]{\frac{x}{h_{1}},1}$, where $h_{0}$ and $h_{1}$ are the asymptotically optimal bandwidths at $z=0$ and $z=1$ and $x$ is the point of evaluation. This bandwidth selection rule implies that the same window is used to estimate the density and its derivatives at all points within $h_{1}$ of the boundary, i.e.\ $x+h = 2h_{1}$ for all $x<h_{1}$.} For the asymptotically unbiased version of our density estimator, PS$_{3}$ with $S=1$, there is no finite optimal bandwidth unless we assume more derivatives of $f$. Instead, we choose the bandwidth at $z=0$ so that the asymptotic variance is the same as the variance of PS$_{2}$ at the boundary. For the generalized reflection method, which requires separate bandwidths and finite--difference approximation to estimate $L'$ in a first step, we do not develop a theory of the asymptotically optimal inputs. Instead, we select a finite--differencing scheme and choose a combination of auxiliary bandwidths so that the pilot estimate of $L'$ at the boundary and the density estimate at $z=1$ have the same asymptotic variances as our method using $g\parens{t}=\parens{t+z}\parens{1-t}$. Specifically, we use a main bandwidth equal to the asymptotically optimal bandwidth for our method at $z=1$, and we estimate $L'\parens{0}$ using $\cparens{\log f_{nh_{L'}}\parens{2h_{L'}}-\log f_{nh_{L'}}\parens{0}}\mathrel/\parens{2h_{L'}}$, where $f_{nh_{L'}}\parens{2 h_{L'}}$ and $f_{nh_{L'}}\parens{0}$ are a kernel and boundary--kernel estimator for the density whose variance is proportional to $1/nh_{L'}$.

Comparing the square root of the mean squared error (RMSE) of the density estimates at zero in \cref{tab:rmsef0} and the RMSE for the derivative in \cref{tab:rmsefprime0}, there is no clear ranking of the estimators. \Cref{fig:rmsef} depicts the bias and RMSE of the local linear estimators for $\theta=4$. The linear approximation of $f$ (LS--CJM) performs well when $f$ is a beta distribution, but has difficulty estimating the truncated normal density and the estimate is often negative. KZ is generally neither the best nor the worst of the estimators we consider here, but we acknowledge that we have not optimized the inputs into the KZ estimator as thoroughly as the other estimator's inputs. We also note that KZ provides a familiar benchmark away from the boundary because it is simply the RP density estimator using an Epanechnikov kernel for $z=1$.

\begin{table}[!htbp] \centering 
  \caption{RMSE of the estimators for the density at the boundary.} 
  \label{tab:rmsef0} 
\begin{tabular}{@{\extracolsep{5pt}} cccccc} 
\\[-1.8ex]\hline 
\hline \\[-1.8ex] 
 & Estimator & $f_1$ & $f_2$ & $f_3$ & $f_4$ \\ 
\hline \\[-1.8ex] 
1 & LS--CJM & $0.067$ & $0.041$ & $0.081$ & $0.050$ \\ 
2 & PS$_1$ & $0.065$ & $0.022$ & $0.068$ & $0.040$  \\ 
3 & PS$_2$ & $0.064$ & $0.022$ & $0.063$ & $0.039$ \\ 
4 & PS$_3$ & $0.062$ & $0.014$ & $0.063$ & $0.041$  \\ 
5 & KZ & $0.077$ & $0.025$ & $0.075$ & $0.042$ \\ 
6 & Loader & $0.064$ & $0.021$ & $0.064$ & $0.039$ \\ 
\hline \\[-1.8ex] 
\end{tabular} 
\end{table}

\begin{table}[!htbp] \centering 
  \caption{RMSE of the estimators for the derivative of the density at the boundary.} 
  \label{tab:rmsefprime0} 
\begin{tabular}{@{\extracolsep{5pt}} cccccc} 
\\[-1.8ex]\hline 
\hline \\[-1.8ex] 
 & Estimator & $f_1$ & $f_{2}$ & $f_3$ & $f_4$ \\ 
\hline \\[-1.8ex] 
1 & LS--CJM & $0.145$ & $0.112$ & $0.394$ & $0.242$  \\ 
2 & PS$_{1}$ & $0.172$ & $0.028$ & $0.436$ & $0.185$  \\ 
3 & PS$_{2}$ & $0.169$ & $0.029$ & $0.400$ & $0.176$  \\ 
4 & PS$_3$ & $0.197$ & $0.054$ & $0.408$ & $0.172$  \\ 
5 & KZ & $0.103$ & $0.153$ & $1.246$ & $0.486$  \\ 
6 & Loader & $0.172$ & $0.029$ & $0.403$ & $0.175$  \\ 
\hline \\[-1.8ex] 
\end{tabular} 
\end{table} 

As expected, PS$_{1}$, PS$_{2}$, PS$_{3}$, and Loader perform similarly away from the boundary. The differences between these estimators in the top two figures are due to the relatively large bandwidth, but the curves are nearly indistinguishable in the bottom two figures. At the boundary, Loader appears to consistently achieve a smaller RMSE than PS$_{1}$ as our asymptotic theory predicts. The difference between PS$_{2}$ and Loader is almost indiscernible in this sample size, though we expect PS$_{2}$ to have a slightly smaller asymptotic bias and variance.

While our estimator can closely mimic the performance the local likelihood estimator, our method can also achieve significantly smaller AMSE if we use a larger bandwidth and an alternative $g$. As an extreme example, the simulations results show the bias of PS$_{3}$ is relatively small in all four designs; in fact, it converges to zero faster than $h_{0}^{2}$.  But its MSE appears to be roughly the same at the boundary as the other estimators' and is generally larger for $z>0$ in three of the four designs, indicating that the finite--sample costs of the asymptotic benefits do not justify this ambitious choice of $g$. The notable exception is in the case of the truncated normal, where the higher--order bias terms are in fact zero due to the fact that $f_{2}$ is log--quadratic. As a result, PS$_{3}$ has a markedly smaller RMSE at the boundary. One could also eliminate the asymptotic bias and significantly reduce the MSE for other values of $z$. For example, in the interior one could use $m_{z}\parens{t}=\frac{3}{4}\parens{1-t}^{2}\parens{1+t}$, $g\parens{t}=\parens{t+z}{1-t}\parens{t^{2}+2t - 1}$, and a bandwidth that is 1.875 times larger than that used for Loader and our other estimators, as suggested in \cref{section:optimalg}. In fact, if the density possesses more derivatives than the researcher was willing to assume, the asymptotic gains will typically be realized more quickly in the interior than at the boundary because the third--order bias term is zero, as well. 

We interpret these simulation results as a proof of concept that our approach can improve on the asymptotically optimal local polynomial and RP estimators without sacrificing continuity or nonnegativity of the estimated density. Moreover, we demonstrate these gains are possible in empirically relevant sample sizes even when the researcher ambitiously attempts to eliminate the asymptotic bias at the boundary. In practice, however, researchers might prefer less extreme versions of our estimator that do not require such large bandwidths to achieve a lower AMSE than commonly used alternatives. Indeed, the asymptotically unbiased version of our estimator does not minimize the AMSE for any bandwidth sequence on the order of $n^{-1/5}$, and it would only be advisable if the researcher specifically requires an unbiased estimate.

\begin{figure}[htbp]
\begin{center}
\includegraphics[width=.49\textwidth]{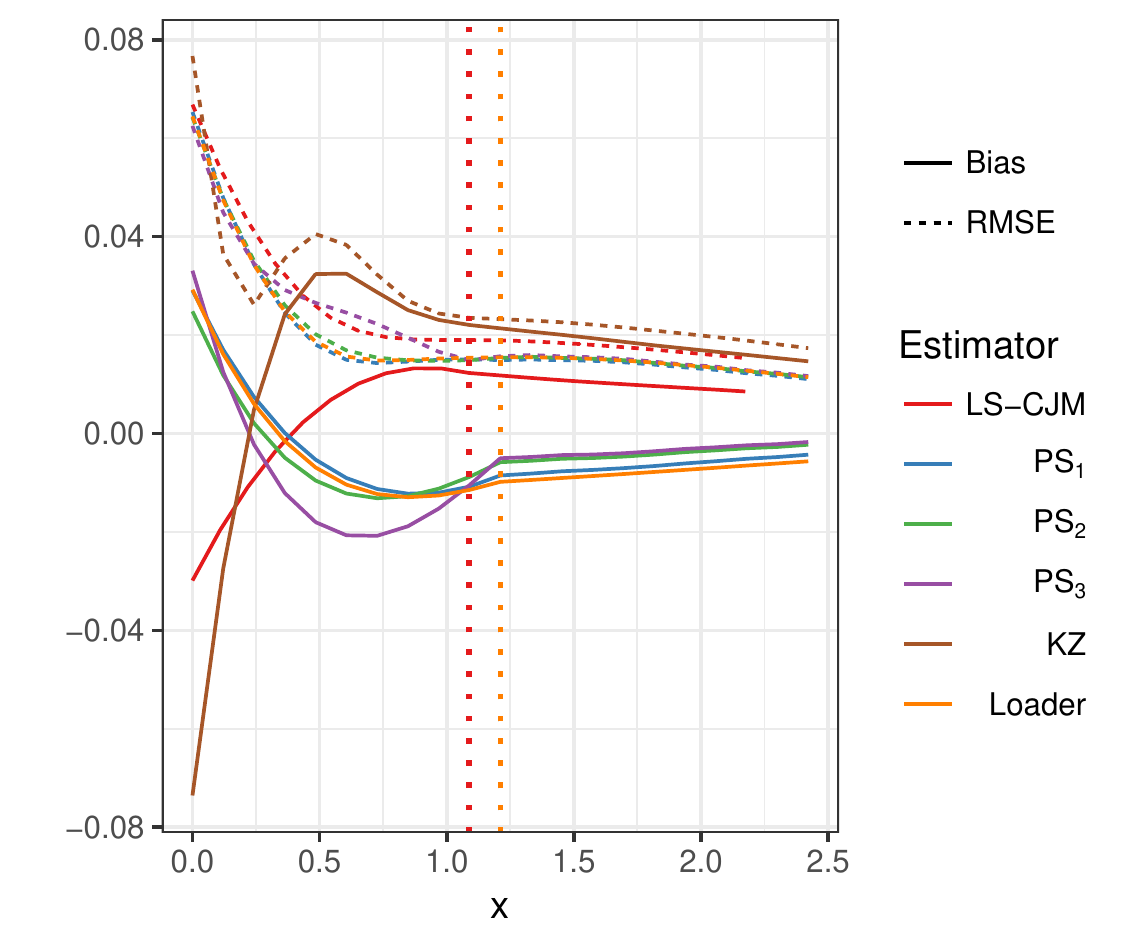}
\includegraphics[width=.49\textwidth]{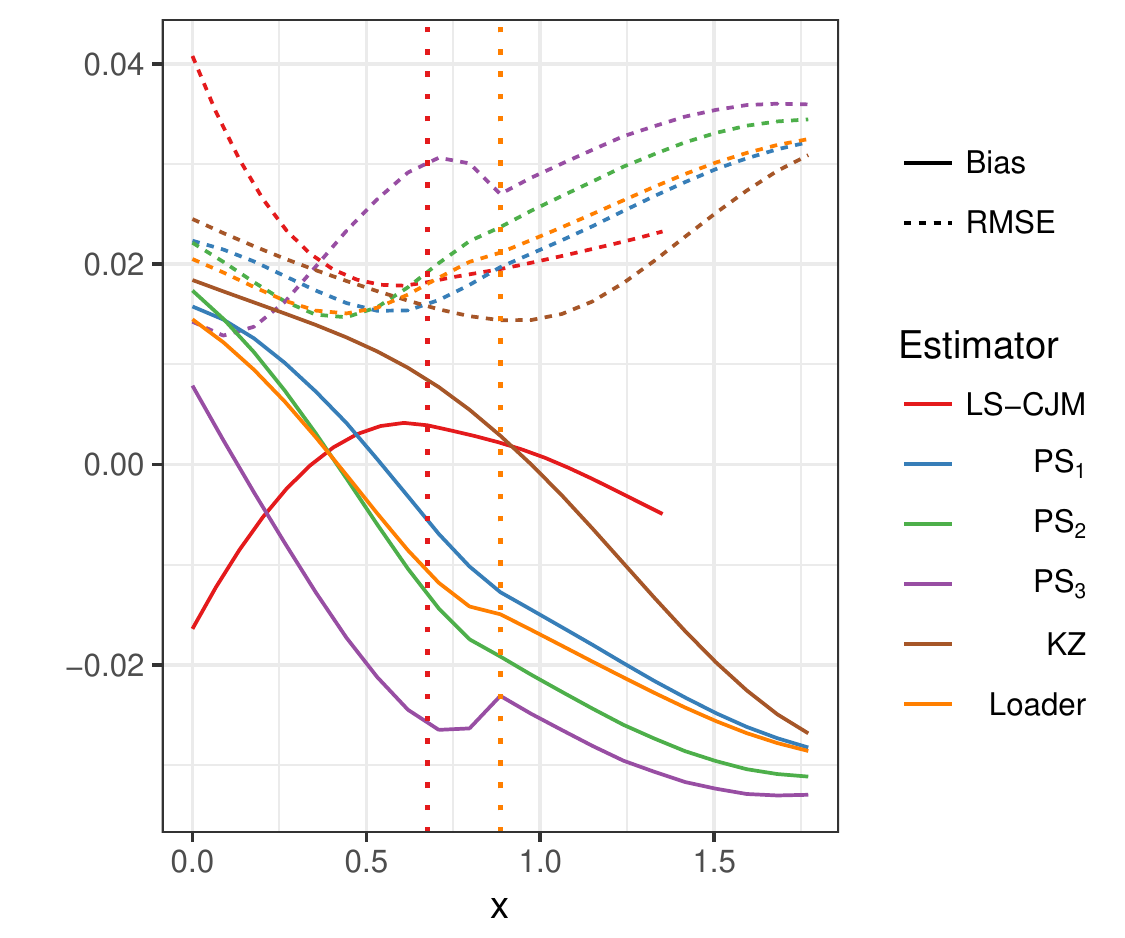}
\includegraphics[width=.49\textwidth]{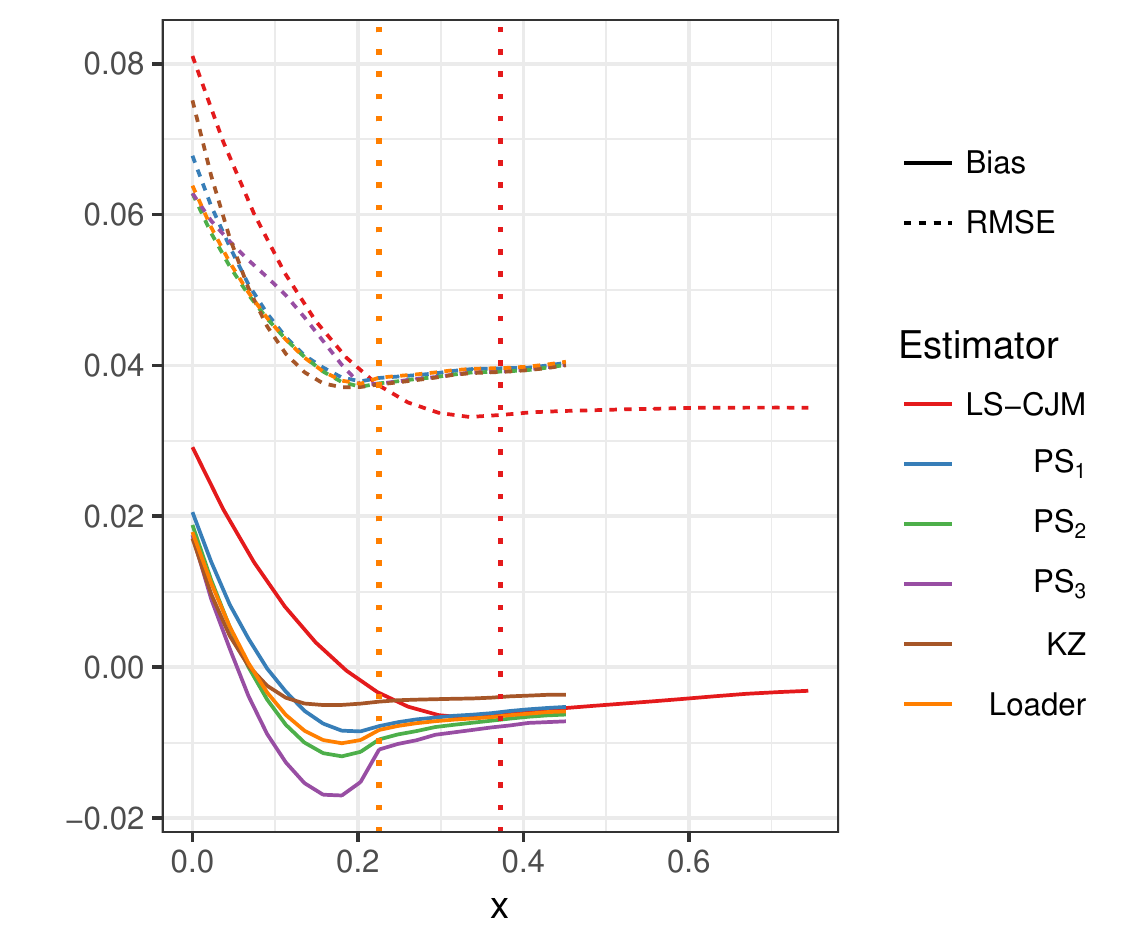}
\includegraphics[width=.49\textwidth]{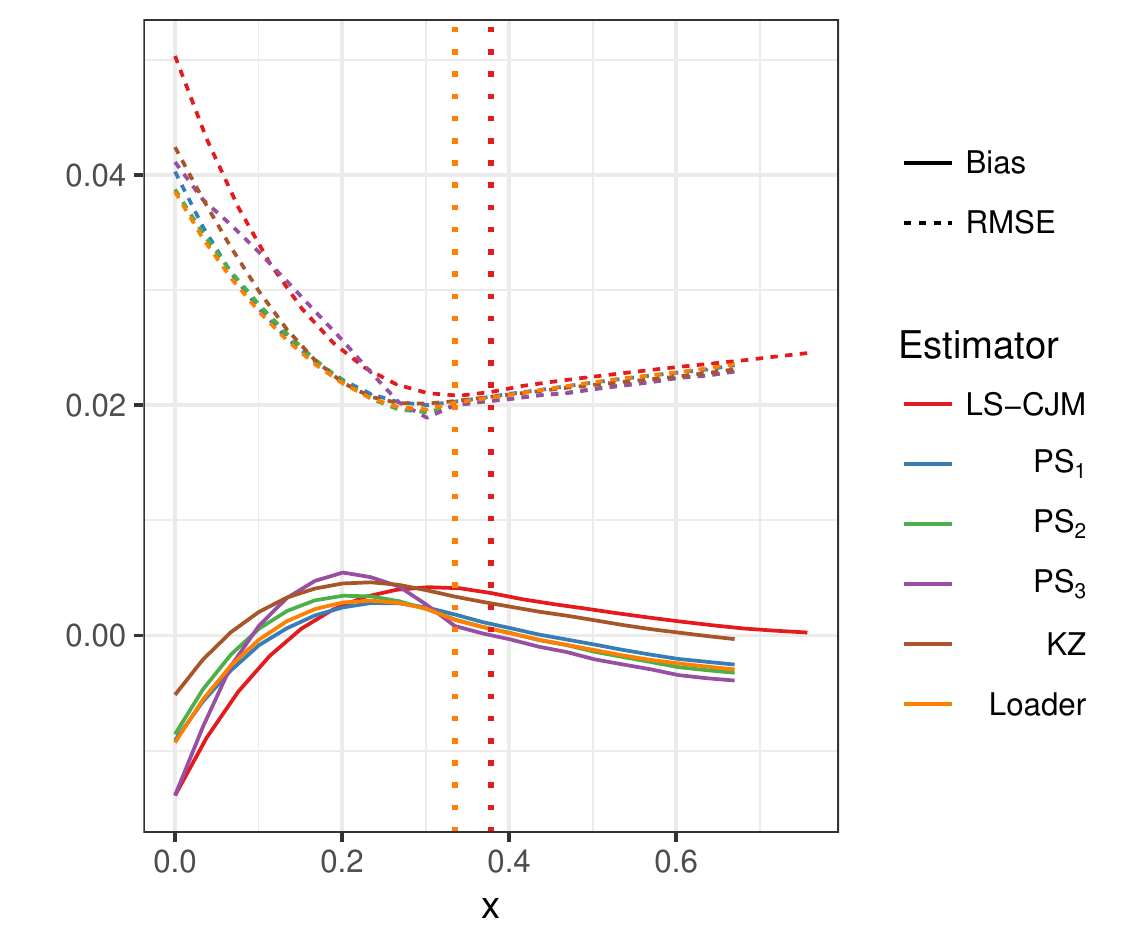}
\caption{The bias and RMSE for estimators of $f_{1}$, $f_{2}$, $f_{3}$, and $f_{4}$ (arranged from left to right) for $\theta=4$. The vertical dotted lines mark the bandwidth used in the interior.}
\label{fig:rmsef}
\end{center}
\end{figure}

In \cref{fig:rmsefprime} we plot the bias and RMSE for the derivative of $f$. In three of the four simulation designs, PS$_{1}$ has the smallest RMSE among the log--linear approximation methods in the interior region, which we expected because $g\parens{t}=\parens{t+z}\parens{t-1}$ minimizes the AMSE in $L'$ using the given bandwidth sequence. At the boundary, however, PS$_{1}$ has a significantly larger AMSE than the alternative choices of $g$, and this is borne out in the simulations to some extent even though the estimator for the derivative converges at the relatively slow rate of $n^{-1/5}$.

\begin{figure}[htbp]
\begin{center}
\includegraphics[width=.49\textwidth]{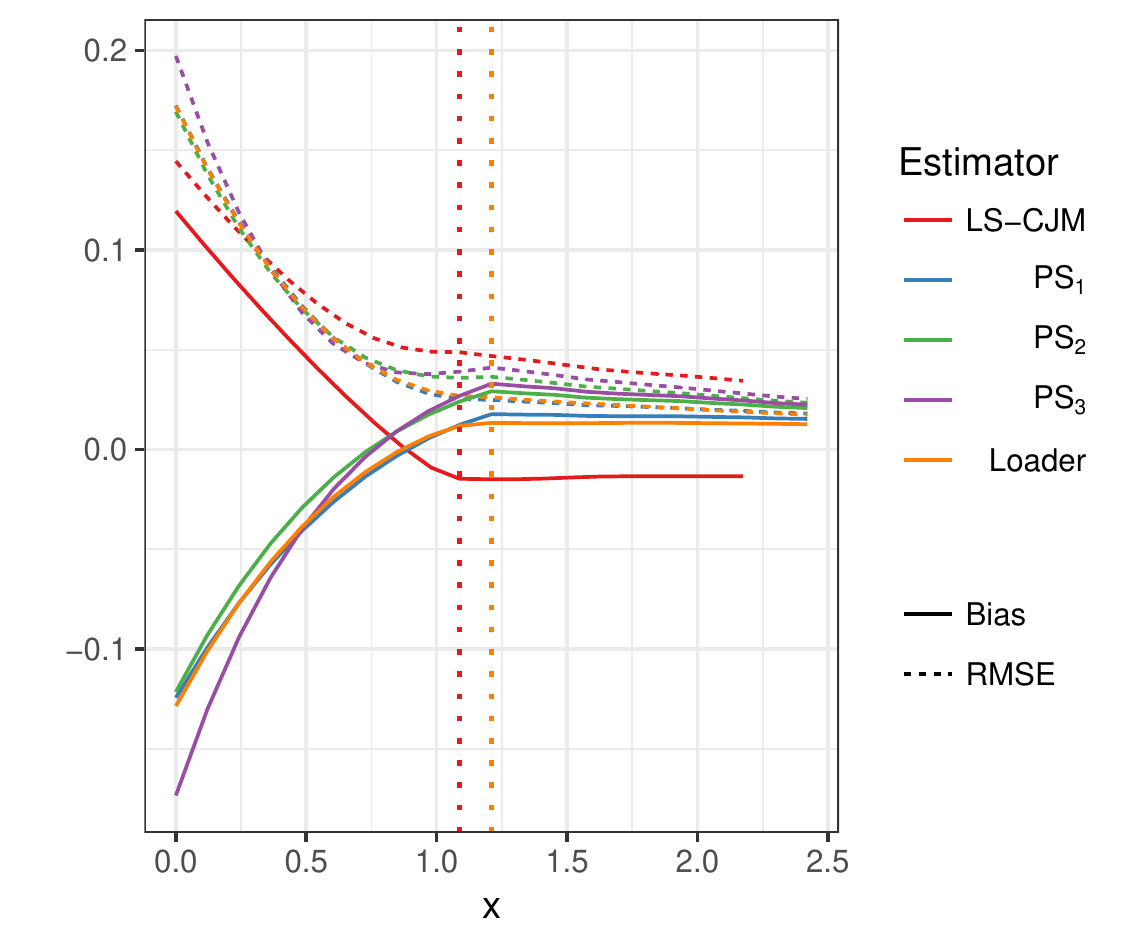}
\includegraphics[width=.49\textwidth]{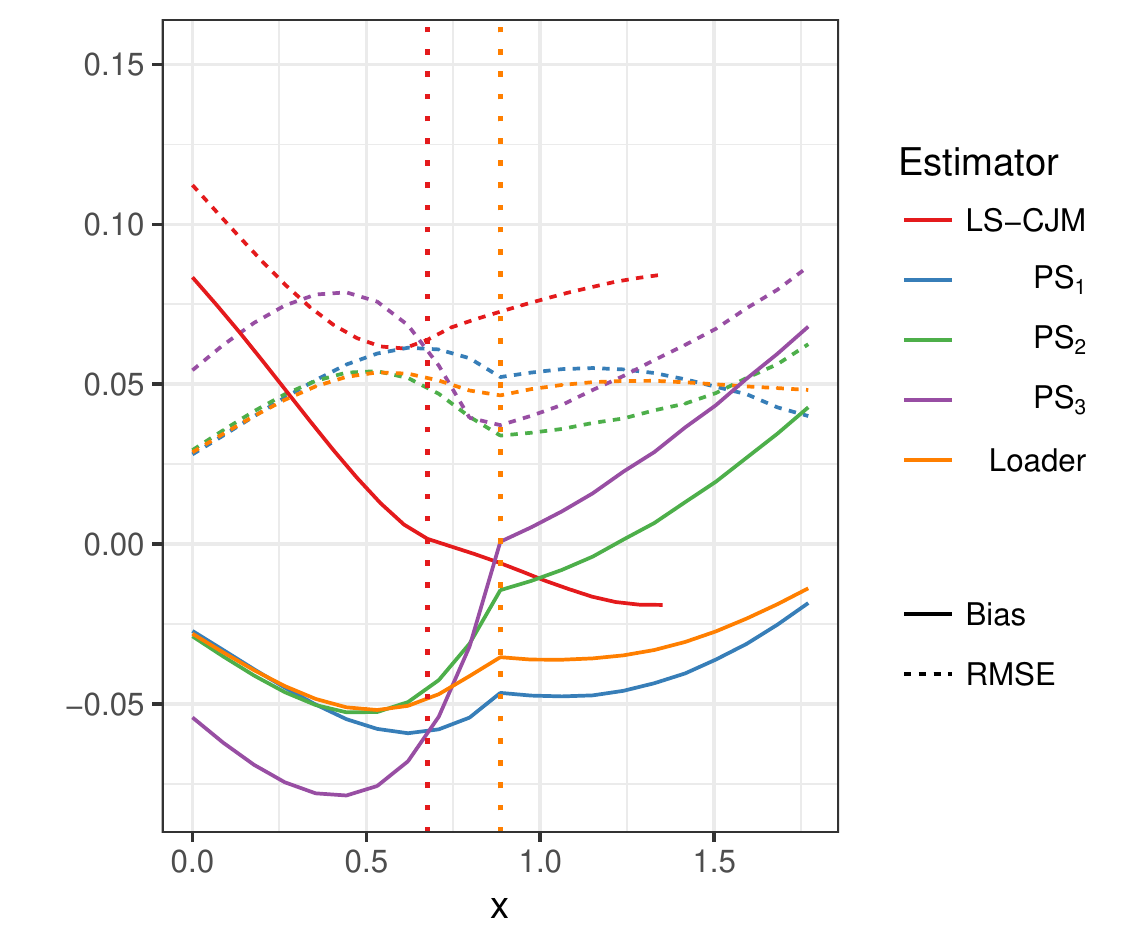}
\includegraphics[width=.49\textwidth]{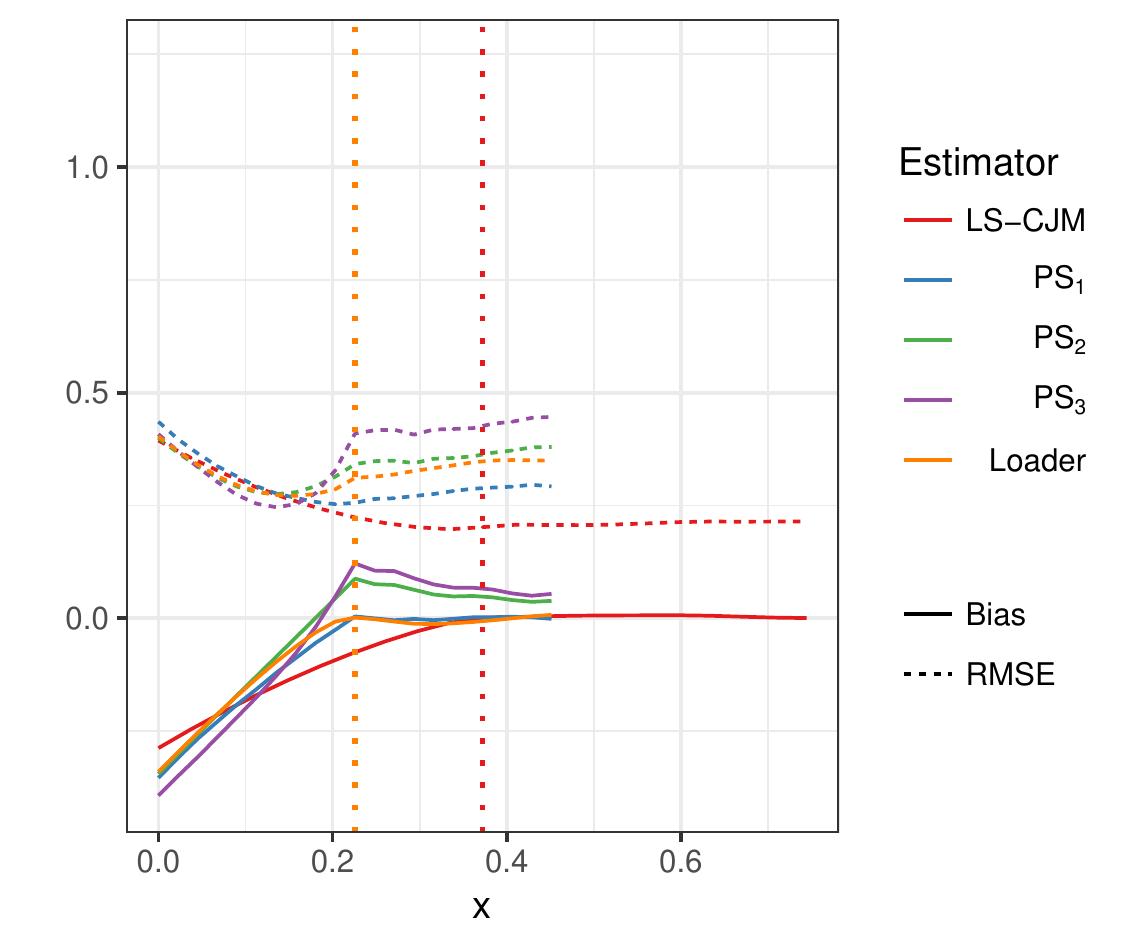}
\includegraphics[width=.49\textwidth]{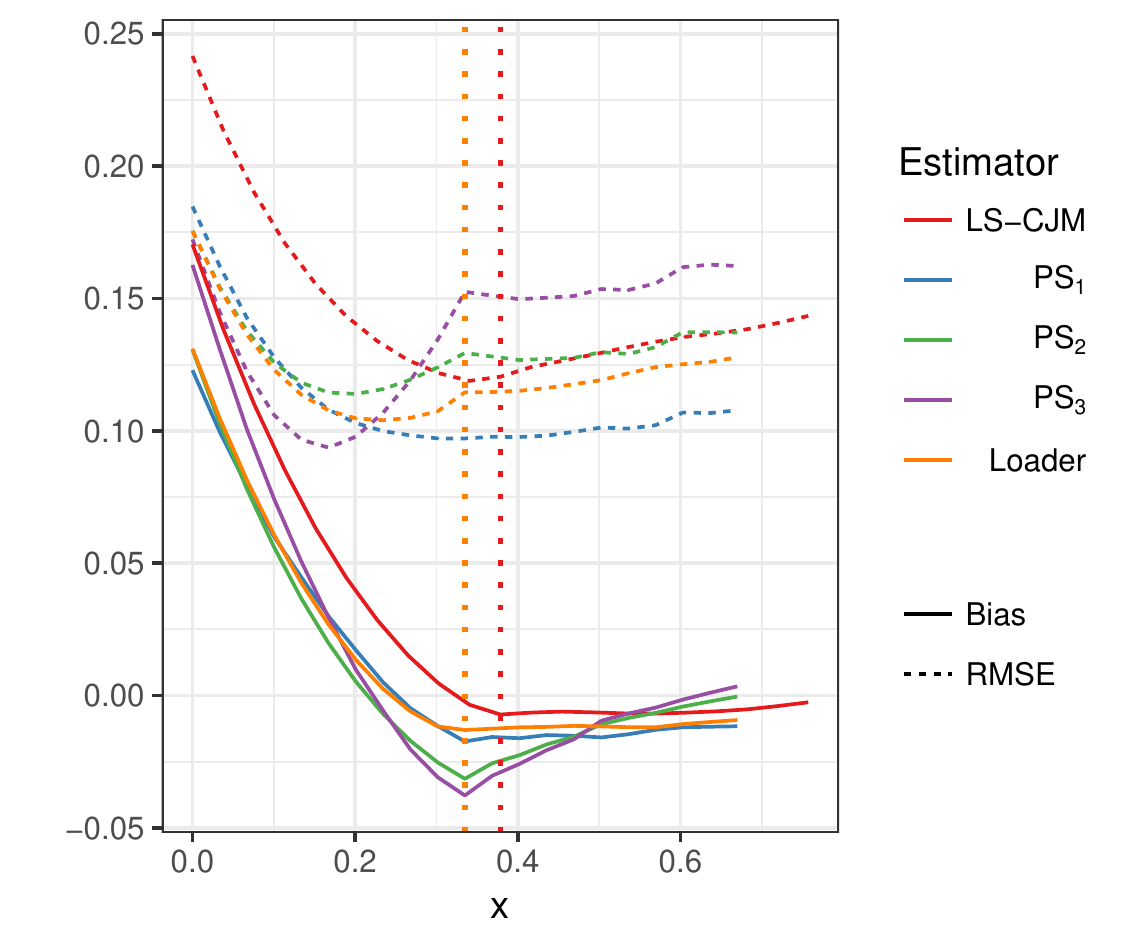}
\caption{The bias and RMSE for estimators of $f'_{1}$, $f'_{2}$, $f'_{3}$, and $f'_{4}$ (arranged from left to right) for $\theta=4$. The vertical dotted lines mark the bandwidth used in the interior.}
\label{fig:rmsefprime}
\end{center}
\end{figure}

\section{Conclusion}
\label{section:conclusion}
We develop an asymptotically normal nonparametric estimator based on a log--polynomial approximation to the unknown density. By approximating the log--density with a polynomial, we can guarantee our estimated density is nonnegative; and by using a polynomial approximation instead of a local constant approximation, we achieve the optimal rates of convergence at the boundary of the support as well as in the interior.

Because our approach allows for a relatively larger degree of customization---the researcher must specify a bandwidth, a kernel, and a vector--valued function $g$ that is zero at the extremes of its support---we explore the optimal set of inputs.  Unlike the standard analysis of optimal kernel and bandwidth inputs, our estimator is nonnegative and continuous in the point of evaluation under a relaxed set of constraints. Because these constraints were needed in order to derive the optimal kernel for use with alternative methods, there is no interior solution to the optimal choice of inputs using our approach. If one fixes the bandwidth sequence, however, the choice of kernel and $g$ can be optimized in a straightforward manner, and our method can achieve the same asymptotic variance as the optimized alternatives if one uses the same bandwidth. Moreover, if the researcher uses a slightly larger bandwidth, marginal reductions in the asymptotic mean squared error are possible.

In fact, there is no positive lower bound to the asymptotic mean squared error using our approach if the researcher is willing to adopt a large enough bandwidth sequence. Although such large bandwidths stretch the plausibility of the asymptotic approximation to the mean squared error in finite samples, we provide simulation evidence that demonstrates significant reductions in the mean squared error can be realized in empirically relevant sample sizes.

More generally, our approach is based on a sample analog to partial integration and can be applied to other settings, as well, such as the estimation of hazard functions or propensity scores.

\appendix

\section{Proofs}

\subsection{Definitions}
\label{app:definitions}

To simplify the argument of $g_j$, define $g_{hxj}(y) = g_j\parens[\big]{\frac{y-x}{h}}$, noting that $g'_{hxj}(y) = \frac1h g_j'\parens[\big]{\frac{y-x}h}$. Rewrite \cref{eq:population system} as 
\(
\ell_{hx} = R_{hx} \beta + r_{hx},
\)
where $\ell_{hx},R_{hx}$ are a vector and a matrix, whose elements are
\< \label{eq:ell R def}
\maligned{
	\ell_{hxj} &=  \Expcr[\bigg]{ g_{hxj}'\parens{\R x_1} \given x-zh \leq \R x_1 \leq x+h},
	\\
	R_{hxjs} & =  - \Expcr[\bigg]{ g_{hxj}\parens{\R x_1} \frac{\parens{\R x_1 - x}^{s-1}}{\parens{s-1}!} \given x-zh \leq \R x_1 \leq x+h},
}
\qquad j,s=1,2,\dots,
\>
and where $r_{hx}=\ell_{hx}-R_{hx}\beta$ is the vector with elements
\< \label{eq:r def}
r_{hjx} = \ell_{hjx} - \sum_{s=1}^S R_{hjsx} \beta_s = R_{hj,S+1,x} \beta_{S+1} + o\parens[\big]{h^S}.
\>
Let $\lh_{hx},\Rh_{hx}$ be sample analogs of $\ell_{hx},R_{hx}$, e.g.
\[
 \lh_{hxj} = \frac{\dsum_{i=1}^n g_{hxj}'\parens{\R x_i}}{n \cparens{\Fhat\parens{x+h}-\Fhat\parens{x-zh}}}, 
\] 
where $\Fhat$ is the empirical distribution function.

Define
\(
\xi_{js}\parens{t} = -\cparens[\big]{g_j(t) t^{s-1}/\parens{s-1}! = \parens{z+t}^j \parens{1-t} t^{s-1} / \parens{s-1}!}\one\parens{-z \leq t \leq 1}
\)
\ and 
\(
\tilde \beta = R_{hx}^{-1} \ell_{hx}.
\)

\subsection{Estimation of derivatives}

\subsubsection{Bias}

\begin{lem} \label{lem:Omega}
\(
 \Omega_{js} = \int_{-z}^1 \xi_{jsz}\parens{t} \dif t = \frac1{\parens{s-1}!} \int_{-z}^1 \parens{z+t}^j \parens{1-t} t^{s-1} \dif t.
\)	
\end{lem}
\begin{proof}
The right hand side equals
\begin{multline*}
\frac1{\parens{s-1}!} \int_{-z}^1 \parens{z+t}^j \cparens{1+z -\parens{z+t}} \parens{z+t-z}^{s-1} \dif t =
\\
\frac1{\parens{s-1}!} \parens[\bigg]{ 
	\parens{1+z}\int_{-z}^1 \parens{z+t}^j \parens{z+t-z}^{s-1} \dif t 
	- \int_{-z}^1 \parens{z+t}^{j+1} \parens{z+t-z}^{s-1}  \dif t
} =
\\
\frac1{\parens{s-1}!} \sum_{i=0}^{s-1} {{s-1}\choose i} \parens{-z}^i
\parens[\bigg]{
\parens{1+z} \int_{-z}^1 \parens{z+t}^{j+s-1-i} \dif t
- 
\int_{-z}^1 \parens{z+t}^{j+s-i} \dif t
}
=
\\
 \sum_{i=0}^{s-1}\frac{\parens{-z}^i \parens{1+z}^{j+s+1-i}}{i! \parens{s-1-i}!\parens{j+s-i}\parens{j+s-i+1}}
 =
 \sum_{i=0}^{s-1} \sum_{t=0}^i {i\choose t} \frac{ \parens{-1}^{i+t} \parens{1+z}^{j+s+1-t} }
  {i! \parens{s-1-i}! \parens{j+s-i}\parens{j+s-i+1}}
  \\
  =
  \sum_{t=0}^{s-1} \parens{1+z}^{j+s+1-t} \sum_{i=t}^{s-1}  \frac{ \parens{-1}^{i+t}  }
  {t!\parens{i-t}! \parens{s-1-i}! \parens{j+s-i}\parens{j+s-i+1}}
  =
  \sum_{t=0}^{s-1} \parens{1+z}^{j+s+1-t}  \frac{\parens{-1}^{s+1+t}\parens{s-t} j! }{\parens{j+s+1-t}!t!}. \qedhere 
\end{multline*}
\end{proof}

\begin{lem} \label{lem:r}
\(
 r_{hjx} = \ell_{hjx} - \sum_{s=1}^S \beta_s R_{hjsx} = \beta_{S+1} h^S \Omega_{j,S+1} / \parens{1+z} + o\parens{h^S}.
\)
\end{lem}
\begin{proof}
From \cref{eq:ell R def,eq:r def} it follows that 
\begin{multline*}
 r_{hxj} = \ell_{hxj} - \sum_{s=1}^S \beta_s R_{hxjs} = - \frac1{F\parens{x+h}-F\parens{x-zh}} \int_{x-zh}^{x+h}  g_{hxj}\parens{y} L'\parens{y} \dif y- \sum_{s=1}^S \beta_s R_{hjsx}
 =\\
 -
  \frac{\beta_{S+1}}{F\parens{x+h}-F\parens{x-zh}} \int_{x-zh}^{x+h}  g_{hxj}\parens{y} 
   \frac{ \parens{y-x}^S}{S!}  f\parens{y} \dif y + o\parens{h^S}
   \\
  = 
  -  
  \frac{h^{S+1} \beta_{S+1}}{F\parens{x+h}-F\parens{x-zh}} \int_{-z}^1  \cparens[\big]{\parens{t+z}^{j+1} - \parens{1+z} t^j }
  \frac{ t^S}{S!}  f\parens{x+th} \dif t + o\parens{h^S}
  \\
  =
  - \frac{h^{S} \beta_{S+1}}{f\parens{x} \parens{1+z}} \int_{-z}^1 \xi_{j,S+1,z}\parens{t} f\parens{x} \dif t + o\parens{h^{S}}
  \\
  =
  \frac{h^S \beta_{S+1} \Omega_{j,S+1}\parens{z}}{1+z} + o\parens{h^S},
\end{multline*}
by the mean value theorem and \cref{lem:Omega}.
\end{proof}

\begin{lem} \label{lem:R}
\(
 R_{hxjs} = h^{s-1}  \Omega_{js} / \parens{1+z} + o\parens{h^{s-1}}.
\)	
\proof
Repeat the steps of the proof of \cref{lem:r}. \qed
\end{lem}

\subsubsection{Distribution}

Let $\Bh_{hxjs} = \Rh_{hxjs} \cparens{\Fhat\parens{x+h}-\Fhat\parens{x-zh}} /h^s$,
$B_{hxjs} = R_{hxjs} \cparens{F\parens{x+h}-F\parens{x-zh}}/h^s$,
$\ah_{hx} = \lh_{hx}\cparens{\Fhat\parens{x+h}-\Fhat\parens{x-zh}}/h$, and
$a_{hx} = \ell_{hx} \cparens{F\parens{x+h}-F\parens{x-zh}}/h$.

\begin{lem} \label{lem:B}
\(
 B_{hxjs} = f\parens{x} \Omega_{js}  + o\parens{1}.
\)
\end{lem}
\begin{proof}
Follows from \cref{lem:R}.
\end{proof}

\begin{lem} \label{lem:ugh}
 \(
 \Bh_{hx} - B_{hx} = O_p\parens{1/\sqrt{nh}}
 \) 
 and 
 \(
 \sqrt{nh^3} \parens{\ah_{hx} - a_{hx} } \convd N\cparens{0,f\parens{x} \parens{1+z}V }.
 \)
 \end{lem}
\begin{proof}
We first show the result for $\ah_{hx}$.
\< \label{eq:ugh}
\sqrt{nh^3} \parens{ \ah_{hx} - a_{hx}}
=
\frac1{\sqrt{nh}} \sum_{i=1}^n \cparens[\big]{g_{hxj}'\parens{x_i}-\Exp g_{hxj}'\parens{\R x_1} }
\convd 
N\cparens{0,f\parens{x}V}\,.
\>
The normality of the limit follows from a standard central limit theorem, e.g.\ \citet{eicker1966multivariate}. Because the estimator is linear and $\Set{x_i}$ is i.i.d., the asymptotic mean and variance are easily obtained, e.g. the $(j,s)$--element of the asymptotic variance matrix is 
\[
nh^3\Cov\parens[\big]{\ah_{hxj},\ah_{hxs}} = h\Cov\parens[\big]{g'_{hxj}\parens{\R x_1}, g'_{hxs}\parens{\R x_{1}}}= f\parens{x} \int_{-z}^1 g_j'\parens{t}g_{s}'\parens{t}\dif t + o\parens{1}
= f\parens{x} V+ o\parens{1}.
\]
Now, for any $j,s=1,\dots,S$,
\[
\sqrt{nh} \parens{\Bh_{hxjs} - B_{hxjs} }
= \\
\frac1{\sqrt{nh}} \sum_{i=1}^n \cparens[\bigg]{ \xi_{js}\parens[\Big]{\frac{\R x_i-x}h} - 
	  \xi_{js}\parens[\Big]{\frac{\R x_1-x}h} } ,
\]
which by standard kernel estimation theory has a limiting mean zero normal distribution and is hence $O_p\parens{1}$.
\end{proof}

\begin{prooft}{\cref{thm:beta}}
Recall that $\bt = R_{hx}^{-1}\ell_{hx}$.  The bias result then follows immediately from \cref{lem:r,lem:R}.  Now the asymptotic distribution.
Note that $\bt = \Lambda_h^{-1} B_{hx}^{-1} a_{hx}$.
Likewise $\bh = \Lambda_h^{-1} \Bh_{hx}^{-1} \ah_{hx}$.  Thus, 
\begin{multline} \label{eq:nasty expansion}
\sqrt{nh^3} \Lambda_h \parens{\bh - \bt} =
 \sqrt{nh^3} \parens{\Bh_{hx}^{-1} - B_{hx}^{-1}} \parens{\ah_{hx} - a_{hx}} + \sqrt{nh^3} \parens{\Bh_{hx}^{-1}-B_{hx}^{-1}}a_{hx} 
 + 
 \sqrt{nh^3} B_{hx}^{-1}\parens{\ah_{hx} -a_{hx} }
 =
 \\
 O_p\cparens[\big]{ \parens{nh}^{-1/2}} + O_p\parens{h} + B_{hx}^{-1}\sqrt{nh^3}\parens{\ah_{hx} -a_{hx} }
 =
 B_{hx}^{-1}\sqrt{nh^3}\parens{\ah_{hx} -a_{hx} } + o_p\parens{1}.
\end{multline}
Since $\lim_{n\to\infty} B_{hx} = f\parens{x} \Omega_{js}$ by \cref{lem:B}, the stated result follows from \cref{lem:ugh}. 
\end{prooft}

\subsection{Estimation of $L$}

Let $\psi_n$ denote the convergence rate of $\max_{s=1,\dots,S} \abs[\big]{ h^s \parens{\bh_s -\beta_s} }$, i.e.\ $n^{-\parens{S+1}/\parens{2S+3}}$, which is $n^{-2/5}$ if $S=1$.  If $S=2$  then the rate is $n^{-3/7}$.

Call the numerator and denominator in \cref{eq:fhat def} $\Nhat$ and $\Dhat$, respectively, and let $N=f\parens{x}$, $D=1$.  Then, our proofs will be based on an expansion of the form
\< \label{eq:ND expansion}
 \frac\Nhat\Dhat - \frac ND \simeq \frac{ \parens{\Nhat-N} - \parens{N/D} \parens{\Dhat -D}  }D  \simeq
  \cparens{\Nhat -f\parens{x}} - f\parens{x} \parens{\Dhat -1 },
\>
where $\simeq$ means that the omitted terms are asymptotically negligible.

The first lemma is concerned with addressing the influence of the $\bh$'s from $\Dhat$.
\begin{lem} \label{lem:exponential integrals approximation}
	\<\label{eq:exponential integrals approximation}
 \dint_{-z}^1  m_z\parens{t} \cparens[\bigg]{\exp\parens[\Big]{\dsum_{s=1}^S \dfrac{\bh_s t^sh^s}{s!}} - \exp\parens[\Big]{\dsum_{s=1}^S \dfrac{\beta_s t^sh^s}{s!}} } \dif t
=
\sum_{s=1}^S h^s\parens{\bh_s-\beta_s} c_{msz} + o_p\parens{\psi_n},
	\>
\end{lem}
\begin{proof}
Expanding $\exp\parens[\big]{\sum_{s=1}^S b_s h^s t^s /s!}$ about $\sparens{\begin{matrix}\beta_1 h & \cdots &\beta_S h^S\end{matrix}}$, the left side in \cref{eq:exponential integrals approximation} equals
\begin{multline*} 
\int_{-z}^1 m_z\parens{t}\sum_{s=1}^S \parens{\bh_s - \beta_s} \frac{t^s h^s}{s!} \exp\parens[\bigg]{\frac{\sum_{s=1}^S \beta_s t^s h^s}{s!}} +o_p\parens{\psi_n}
= \sum_{s=1}^S \parens{\bh_s - \beta_s}h^s \int_{-z}^1 \frac{m_z\parens{t}t^s}{s!}\frac{f(x + th)}{\exp\parens{\beta_0}}\dif t + o_p\parens{\psi_n}\\
	= \sum_{s=1}^S \parens{\bh_s - \beta_s}h^s \int_{-z}^1  \frac{m_z\parens{t}t^s}{s!}\frac{f(x)}{\exp\parens{\beta_0}}\dif t + o_p\parens{\psi_n}
		= \sum_{s=1}^S \parens{\bh_s - \beta_s}h^s c_{msz}+ o_p\parens{\psi_n},
\end{multline*}
as asserted.
\end{proof}

\subsubsection{Bias}

The next few lemmas are concerned with the asymptotic bias.

\begin{lem} \label{lem:L num bias}
	The bias in $\Nhat-f\parens{x}$ is
	\(
	f\parens{x} \sum_{s=1}^{S+1} \sP_s\parens{\beta_1,\dots,\beta_s} c_{msz}h^s + o\parens{h^{S+1}},
	\)
	where $\sP_s$ is a complete exponential Bell polynomial.
\end{lem}
\begin{proof}
	We have
\begin{multline*}
	\Exp\Nhat = \frac1h \int_{x-zh}^{x+h} m_z\parens[\Big]{\frac{y-x}h} f\parens{y} \dif y =
	\int_{-z}^1 m_z\parens{t} f\parens{x+th} \dif t
	=
	f\parens{x} + \sum_{s=1}^{S+1}  f^{\parens{s}}\parens{x} c_{msz} h^s + o\parens{h^{S+1}}
	\\
	=
	f\parens{x} + f\parens{x} \sum_{s=1}^{S+1} h^s \sP_s\parens{\beta_1,\dots,\beta_s} c_{msz}+ o\parens{h^{S+1}},
\end{multline*}
by Fa\`a di Bruno's theorem.
\end{proof}

\begin{lem} \label{lem:L den bias}
	Letting $\beta_s^* = \one\parens{s\leq S} \beta_s$, the bias in $\Dhat - 1$ is
\(
 \sum_{s=1}^{S+1}  \sP_s\parens{\beta_1^*,\dots,\beta_s^*} c_{msz} h^s + \beta_{S+1 }h^{S+1} \sum_{s=1}^S c_{msz} b_s  + o\parens{\psi_n}.
\)
\end{lem}
\begin{proof}
The bias has two components: one is due to the estimation of $\beta$ and the other to the finite polynomial approximation.  Consider the latter first.  We have
\begin{multline*}
	\frac1h \int_{x-zh}^{x+h} m_z\parens[\Big]{\frac{y-x}h} \exp\parens[\Big]{ \sum_{s=1}^S \frac{\beta_s \parens{y-x}^s}{s!} } \dif y - 1
	=
	\int_{-z}^1 m_z\parens{t} \exp\parens[\Big]{ \sum_{s=1}^\infty \frac{ \beta_s^* t^s h^s}{s!}} \dif t -1
	=
	\\
	 \sum_{s=1}^{S+1}  \sP_s\parens{\beta_1^*,\dots,\beta_s^*} c_{msz} h^s + o\parens{h^{S+1}}.
\end{multline*}	
For the bias due to the estimation of $\beta$ note that by \cref{lem:exponential integrals approximation,thm:beta} this bias is equal to\\
\(
	 \sum_{s=1}^S h^s c_{msz} \parens{ \beta_{S+1} h^{S+1-s}  b_s } + o\parens{\psi_n}
	 =
	 \beta_{S+1 }h^{S+1} \sum_{s=1}^S c_{msz} b_s + o\parens{\psi_n}.
\)
\end{proof}
	
\begin{lem} \label{lem:L bias dominant term}
The bias in $\Nhat - f\parens{x} \Dhat$ is 
\(
f\parens{x} \beta_{S+1} h^{S+1} \parens[\big]{ c_{m,S+1,z} - \sum_{s=1}^S c_{msz} b_s} + o\parens{\psi_n}.
\)
\end{lem}
\begin{proof}
By \cref{lem:L num bias,lem:L den bias}, the desired bias is
\begin{multline*}
f\parens{x} \sum_{s=1}^{S+1} \sP_s\parens{\beta_1,\dots,\beta_s} c_{msz}h^s - f\parens{x} \parens[\bigg]{\sum_{s=1}^{S+1}  \sP_s\parens{\beta_1^*,\dots,\beta_s^*} c_{msz} h^s + \beta_{S+1 }h^{S+1} \sum_{s=1}^S c_{msz} b_s} + o\parens{\psi_n}
\\
=
f\parens{x} \beta_{S+1} h^{S+1} \parens[\Big]{ c_{m,S+1,z} - \sum_{s=1}^S c_{msz} b_s} + o\parens{\psi_n},
\end{multline*}
as claimed. 
\end{proof}

\begin{lem} \label{lem:Dhat}
$\Dhat -1 = o_p\parens{1}$.
\end{lem}
\begin{proof}
Follows from \cref{lem:exponential integrals approximation,thm:beta,lem:L den bias}.  
\end{proof}
	
\subsubsection{Distribution}
	
\begin{lem} \label{lem:limit distribution dominant term}
$\sqrt{nh} \cparens[\big]{\Nhat - f\parens{x}\Dhat} \convd N\parens{\sB,\sV }$.	
\end{lem}
\begin{proof}
The bias result follows from \cref{lem:L bias dominant term} and the definition of $\Xi$.  For asymptotic normality and the variance formula, note that 
\[
 \Nhat - \Exp \Nhat = \frac1{nh} \sum_{i=1}^n \cparens[\bigg]{
m_z\parens[\Big]{\frac{\R x_i -x}h} - \Exp m_z\parens[\Big]{\frac{\R x_1 -x}h}.
}
\]
Note further that by \cref{lem:exponential integrals approximation}, the asymptotic distribution of $\Dhat$ net of bias is governed by
\(
  \sum_{s=1}^S h^s\parens{\bh_s-\tilde \beta_s} c_{msz},
\)
which by \cref{eq:nasty expansion,eq:ugh,lem:B} is
\[
 h c_{mz}^\tr \Lambda_h \parens{\bh - \tilde \beta} =
 \frac{c_{mz}^\tr \Omega^{-1}}{f\parens{x}}  \frac1{nh} \sum_{i=1}^n 
 \cparens[\bigg]{
g'\parens[\Big]{ \frac{\R x_i-x}h} - \Exp   g'\parens[\Big]{ \frac{\R x_1-x}h} 
} + o_p\parens{\psi_n}.
\]
Thus,
\begin{multline*}
 \sqrt{nh} \cparens{\Nhat - f\parens{x} \Dhat - \Exp\parens{\cdot}} 
 =
 \frac1{\sqrt{nh}} \sum_{i=1}^n \cparens[\bigg]{
 	m_z\parens[\Big]{\frac{\R x_i -x}h} -  c_{mz}^\tr \Omega^{-1}  g'\parens[\Big]{ \frac{\R x_i-x}h}
 	-
 \Exp\parens\cdot  
 } + o_p\parens{1}
\\
=
\frac1{\sqrt{nh}} \sum_{i=1}^n \cparens[\bigg]{
	\omega_z\parens[\Big]{\frac{\R x_i -x}h} -  \Exp 	\omega_z\parens[\Big]{\frac{\R x_1 -x}h} } + o_p\parens{1},
\end{multline*}
which has the stated limit distribution by e.g.\ \citet{eicker1966multivariate}.
\end{proof}

\begin{lem} \label{lem:L omitted terms}
The omitted terms in \cref{eq:ND expansion} are $o_p\parens{\psi_n}$. 
\end{lem}
\begin{proof}
Write
\(
\Nhat/\Dhat - f\parens{x}/1 = 
\cparens{\Nhat - f\parens{x} \Dhat} + \cparens{\Nhat - f\parens{x} \Dhat} \parens{1/\Dhat -1}
	\simeq
\Nhat - f\parens{x} \Dhat. 	
\)
by \cref{lem:Dhat}.
Apply \cref{lem:limit distribution dominant term,lem:L bias dominant term}.
\end{proof}

\begin{prooft}{\cref{thm:f}}
We show the result for $\fhat$ where the result for $\Lhat$ follows from the delta method.	
By \cref{lem:L omitted terms} we only need to consider $\Nhat - f\parens{x}\Dhat$.  Apply \cref{lem:limit distribution dominant term,lem:L bias dominant term}.
\end{prooft}

\setstretch{1}
{\small
\bibliographystyle{apa}
\bibliography{boundest.bib}

\begin{thebibliography}{}

\bibitem[\protect\astroncite{Bell}{1927}]{bell1927partition}
Bell, E.~T. (1927).
\newblock Partition polynomials.
\newblock {\em Annals of Mathematics}, pages 38--46.

\bibitem[\protect\astroncite{Cattaneo et~al.}{2019}]{Cattaneo2019simple}
Cattaneo, M.~D., Jansson, M., and Ma, X. (2019).
\newblock {Simple local polynomial density estimators}.
\newblock {\em Journal of the American Statistical Association}, pages 1--7.

\bibitem[\protect\astroncite{Cheng et~al.}{1997}]{Cheng1997auto}
Cheng, M.-Y., Fan, J., and Marron, J.~S. (1997).
\newblock {On automatic boundary corrections}.
\newblock {\em Annals of Statistics}, 25(4):1691--1708.

\bibitem[\protect\astroncite{Eicker}{1966}]{eicker1966multivariate}
Eicker, F. (1966).
\newblock A multivariate central limit theorem for random linear vector forms.
\newblock {\em Annals of Mathematical Statistics}, pages 1825--1828.

\bibitem[\protect\astroncite{Epanechnikov}{1969}]{Epanechnikov1969}
Epanechnikov, V.~A. (1969).
\newblock Nonparametric estimation of a multidimensional probability density.
\newblock {\em Theory of Probability and its Applications}, 14:156--161.

\bibitem[\protect\astroncite{Gasser and M{\"u}ller}{1979}]{gasser1979kernel}
Gasser, T. and M{\"u}ller, H.-G. (1979).
\newblock Kernel estimation of regression functions.
\newblock In {\em Smoothing techniques for curve estimation}, pages 23--68.
  Springer.

\bibitem[\protect\astroncite{Guerre et~al.}{2000}]{Guerre2000}
Guerre, E., Perrigne, I., and Vuong, Q. (2000).
\newblock {Optimal nonparametric estimation of first--price auctions}.
\newblock {\em Econometrica}, 68(3):525--574.

\bibitem[\protect\astroncite{Hickman and Hubbard}{2015}]{hickman2015replacing}
Hickman, B.~R. and Hubbard, T.~P. (2015).
\newblock Replacing sample trimming with boundary correction in nonparametric
  estimation of first-price auctions.
\newblock {\em Journal of Applied Econometrics}, 30(5):739--762.

\bibitem[\protect\astroncite{Hirano et~al.}{2003}]{hirano2003efficient}
Hirano, K., Imbens, G.~W., and Ridder, G. (2003).
\newblock Efficient estimation of average treatment effects using the estimated
  propensity score.
\newblock {\em Econometrica}, 71(4):1161--1189.

\bibitem[\protect\astroncite{Hjort and Jones}{1996}]{Hjort1996local}
Hjort, N.~L. and Jones, M.~C. (1996).
\newblock Locally nonparametric density estimation.
\newblock {\em Annals of Statistics}, 24(4):1619--1647.

\bibitem[\protect\astroncite{Jones and Foster}{1996}]{jones1996simple}
Jones, M. and Foster, P. (1996).
\newblock A simple nonnegative boundary correction method for kernel density
  estimation.
\newblock {\em Statistica Sinica}, pages 1005--1013.

\bibitem[\protect\astroncite{Karunamuni and
  Alberts}{2005}]{Karunamuni2005boundary}
Karunamuni, R.~J. and Alberts, T. (2005).
\newblock {On boundary correction in kernel density estimation}.
\newblock {\em Statistical Methodology}, 2(3):191--212.

\bibitem[\protect\astroncite{Karunamuni and Zhang}{2008}]{Karunamuni2008some}
Karunamuni, R.~J. and Zhang, S. (2008).
\newblock {Some improvements on a boundary corrected kernel density estimator}.
\newblock {\em Statistics and Probability Letters}, 78(5):499--507.

\bibitem[\protect\astroncite{Klein and Spady}{1993}]{klein1993efficient}
Klein, R.~W. and Spady, R.~H. (1993).
\newblock An efficient semiparametric estimator for binary response models.
\newblock {\em Econometrica}, pages 387--421.

\bibitem[\protect\astroncite{Lejeune and Sarda}{1992}]{Lejeune1992smooth}
Lejeune, M. and Sarda, P. (1992).
\newblock {Smooth Estimators of Distribution and Density Functions}.
\newblock {\em Computational Statistics {\&} Data Analysis}, 14:457--471.

\bibitem[\protect\astroncite{Lewbel and Schennach}{2007}]{lewbel2007simple}
Lewbel, A. and Schennach, S.~M. (2007).
\newblock A simple ordered data estimator for inverse density weighted
  expectations.
\newblock {\em Journal of Econometrics}, 136(1):189--211.

\bibitem[\protect\astroncite{Loader}{1996}]{Loader1996likelihood}
Loader, C.~R. (1996).
\newblock {Local likelihood density estimation}.
\newblock {\em Annals of Statistics}, 24(4):1602--1618.

\bibitem[\protect\astroncite{Muller}{1984}]{Muller1984}
Muller, H.-G. (1984).
\newblock Smooth optimum kernel estimators of densities, regression curves and
  modes.
\newblock {\em Annals of Statistics}, 12(2):766--774.

\bibitem[\protect\astroncite{Pinkse and Schurter}{2019}]{pinkse2019estimation}
Pinkse, J. and Schurter, K. (2019).
\newblock Estimation of auction models with shape restrictions.
\newblock Pennsylvania State University working paper.

\bibitem[\protect\astroncite{Zhang and
  Karunamuni}{1998}]{Karunamuni1998endpoints}
Zhang, S. and Karunamuni, R.~J. (1998).
\newblock {On kernel density estimation near endpoints}.
\newblock {\em Journal of Statistical Planning and Inference}, 70:301--316.

\end{thebibliography}
}
\end{document}